\documentclass[journal]{IEEEtran}
\usepackage{amssymb}
\usepackage{amsmath}
\usepackage{bm}
\usepackage{cite}
\usepackage{graphics}
\usepackage{graphicx}
\usepackage{algorithmicx}
\usepackage{algpseudocode}
\usepackage{program}
\usepackage{algorithm}
\usepackage{array}
\usepackage{multirow}
\newcolumntype{P}[1]{>{\centering\arraybackslash}p{#1}}
\newtheorem{thm}{\bf Theorem}

\begin{document}

\title{Pareto Optimal Demand Response Based on Energy Costs and Load Factor in Smart Grid}

\author{Wei-Yu~Chiu,~\IEEEmembership{Member,~IEEE}, Jui-Ting Hsieh, and Chia-Ming Chen

\thanks{This work was supported by the Ministry of Science and Technology of Taiwan under Grant MOST 108-2221-E-007-100. (Corresponding author: Wei-Yu Chiu.)}
\thanks{W.-Y. Chiu and C.-M. Chen are with the Department of Electrical Engineering, National Tsing Hua University,
Hsinchu 30013, Taiwan (e-mail: chiuweiyu@gmail.com).}
\thanks{J.-T. Hsieh is with the Department of Electrical Engineering, Yuan Ze University,
Taoyuan 32003, Taiwan.}%

\thanks{\copyright2019 IEEE. Personal use of this material is permitted. Permission from IEEE must be obtained for all other uses, in any current or future media, including reprinting/republishing this material for advertising or promotional purposes, creating new collective works, for resale or redistribution to servers or lists, or reuse of any copyrighted component of this work in other works.}
\thanks{Digital Object Identifier 10.1109/TII.2019.2928520}

}

\maketitle

\begin{abstract}
Demand response for residential users is essential to the realization of modern smart grids.
This paper proposes a multiobjective approach to designing a demand response program that considers the energy costs of residential users and the load factor of the underlying grid. A multiobjective optimization problem (MOP) is formulated and Pareto optimality is adopted. Stochastic search methods of generating feasible values for decision variables are proposed.
Theoretical analysis is performed to show that the proposed methods can effectively generate and preserve feasible points during the solution process,
which comparable methods can hardly achieve.
A multiobjective evolutionary algorithm is developed to solve the MOP, producing a Pareto optimal demand response (PODR) program.
Simulations reveal that the proposed method outperforms the comparable methods in terms of energy costs while producing a satisfying load factor.
The proposed PODR program is able to systematically balance the needs of the grid and residential users.
\end{abstract}


\IEEEpeerreviewmaketitle

\section{Introduction}

\IEEEPARstart{I}{n} existing power grids, power plants usually deliver a unidirectional power flow to customers, converting only one-third of the total energy in their fuel into electricity and wasting the heat produced. Twenty percent of a power grid's generation capacity is often used only to cover peak loads that account for approximately five percent of the time \cite{5357331,4787536}. When natural disasters occur, conventional power grids are unable to resist or self-heal. Next-generation power grids, known as smart grids, have been proposed and designed to replace traditional power grids with the aim of reducing transmission loss, generating electricity more efficiently, and resisting or self-healing after natural disasters. A few governments have adopted proactive policies to popularize smart grids and construct related infrastructure. Smart grids are expected to have higher electricity transmission stability, detect faults and self-heal, allow bidirectional energy and information flow between customers and suppliers, and reliably defend against attacks or natural disasters~\cite{6861946,7744778}.

In a smart grid, the real-time energy consumption profiles of residential users are crucial to power suppliers\cite{7466137}. With this information, suppliers or local trading centers can design a dynamic  pricing scheme that strengthens market mechanisms and encourages users to shift their peak loads~\cite{7094306,6376270}.
Power scheduling can then be achieved through the direct minimization of energy consumption costs~\cite{6670131,6693775,6866904,6575202,6476768,5540263,6316164,7093186}.
The change of demand curves in response to pricing signals is termed demand response.

Common pricing models include real-time pricing, day-ahead pricing, time-of-use pricing, and critical-peak pricing models\cite{6462005}. Regardless of which dynamic pricing scheme is employed, suppliers can lower the cost of power generation if users' peak loads are shifted. This can be achieved through the use of demand response programs.
Generally, shifting peak loads is related to increasing the load factor, which is the reciprocal of the peak-to-average ratio~\cite{4652590,4914742}. However, to fully utilize the grid capacity, the load factor is a more appropriate metric.

 Several studies have incorporated  the concept of load factor or  peak-to-average ratio into power scheduling.
 Game theory approaches have been widely used for designing demand response programs\cite{14Fadlullah,10Mohsenian-Rad,15Eksin,12Nguyen}.
 Although under certain conditions Nash equilibrium strategies can be Pareto optimal\cite{15Stephens},
 it is not always the case: payoffs of game players may be improved simultaneously after a Nash equilibrium is attained\cite{15Ma}.
Stochastic optimization such as genetic algorithms has been applied to attain a satisfactory load factor as well\cite{13Zhao}, yielding
 single-objective optimization problems. Some drawbacks may inherit from such single-objective formulations, such as the need of prior decision making for the tradeoff between objectives\cite{J14}.

 Because demand response is closely related to pricing signals, a demand response program can be derived from a pricing scheme\cite{18VENIZELOU}. Kunwar \emph{et al.}\cite{13Kunwar} proposed an
  area-load based pricing scheme for demand side management. The load factor and  energy cost were jointly optimized.
  Given pricing signals, the proposed methodology could produce a promising demand response program addressing the two objectives.
 In contrast with  single-objective optimization, such a multiobjective approach could avoid, for example, the need of heuristic assignment for weighting coefficients and the need of prior decision making for tradeoffs induced by conflicting objectives.
 There are, however, a few important issues that were not covered by\cite{13Kunwar}.
  First, although the approach considered shiftable loads, they were addressed statistically.
  In this case, no integer or discrete decision variables were involved in optimization, but they are important for explicit control of shiftable loads.
  Second, owing to the employed statistical modelling, the impact of electric vehicles (EVs) was not explicitly examined.
  In general, EVs pose different physical constraints and should be distinguished from ordinary home appliances.
  Third, renewable energy sources (RESs) that have been widely adopted in residential communities were not considered.

 Motivated by\cite{13Kunwar}, we propose a multiobjective approach that addresses the load factor and energy consumption costs of residential users in separate dimensions, leading to
 a multiobjective optimization problem (MOP).  Our system models include EVs and RESs.
 Regarding the MOP, power scheduling profiles serve as the decision variables, consisting of continuous and discrete ones.
 For these two types of decision variables, feasible search mechanisms are developed by exploiting constraint and variable structures.
Relevant analysis is provided to show their effectiveness.
  Solving the MOP yields an approximate Pareto set and a Pareto front. A Pareto optimal solution is then selected to indicate how power loads should be adjusted over time, yielding a Pareto optimal demand response (PODR) program. This program can be offered by utilities to residential users.
 The associated algorithms or technologies are implemented in residential homes. The load factor is related to grid reliability for utilities while the energy cost pertains to the participation of residential users.

The main contributions of this study are summarized as follows: First, we examine the grid load factor and energy costs of residential users using a multiobjective framework, and consider a combination of various types of appliances, an EV, and RESs. This combination has not been fully investigated in the literature when multiobjective optimization is applied for residential power scheduling. Second, we propose a few algorithms that generate feasible values for decision variables, leading to a solution method for our MOP.
Owing to physical constraints involving discrete and continuous variables, this MOP can hardly be solved by conventional multiobjective evolutionary algorithms.
Third, numerical simulations demonstrate that our approach can be beneficial to residential users while achieving a satisfying load factor as compared to comparable methods.
Finally, in addition to numerical evaluations,  the proposed algorithms are theoretically analyzed, providing a rigorous foundation for the proposed PODR program.

The remainder of this paper is organized as follows: Section~\ref{sec_rel} discusses related work. Section~\ref{sec_model} describes  mathematical models of various home appliances, the energy consumption costs of residential users, and the load factor of the power grid. The PODR program is proposed in Section~III. Section~IV presents our simulation results involving comparisons between existing demand response strategies. Finally, the paper concludes in Section~V.

\section{Related Work}\label{sec_rel}

This section examines various methods employed to design demand response programs for residential users.
To facilitate discussions, the section is divided into two subsections.
The first subsection presents existing residential demand response methods involving the concept of the load factor but without  explicitly addressing it as an objective.
This includes methods that have the load factor as an optimization constraint or have  the load factor
numerically analyzed in the simulations. The second subsection investigates  methods that do not involve the concept of the load factor but pertain to  our theme.

\subsection{Methods With Load Factor}

Game theory,  control methods,  mixed integer nonlinear programming, convex optimization, and machine learning methods have been the most promising
  techniques employed to  realize residential demand response involving the concept of the load factor.
When game theory is used to model demand response programs\cite{16Kamyab,15Forouzandehmehr,1607Safdarian,15ZAREEN,16LiMa,17Li,14Safdarian}, two games are typically involved: one for suppliers such as utilities and aggregators, and the other for customers. A demand response program is realized by attaining the Nash equilibrium of the games.
Game theory approaches often involve bidding, but a bidding scheme  for residential demand response can also be realized using other techniques, see, for example, \cite{14Adika}.
Game theory approaches can be further combined with blockchain technologies to improve network security\cite{18NOOR}.

 Control strategies  have been applied for demand response services as well\cite{18Diekerhof}. In this case, state-space representations are mostly employed to model system dynamics. For instance, Luo \emph{et al.}\cite{19Luo} proposed a multistage home energy management system consisting of forecasting,
day-ahead scheduling, and actual operation.
At the day-ahead scheduling stage, a coordinated home energy resource scheduling model constrained by the peak-to-average ratio
was constructed to minimize a one-day
home operation cost. At the actual operation stage, a model
predictive control based operational strategy was proposed
to correct home energy resource operations with the update of real-time information.

Models for residential energy demand have drawn much attention\cite{16Pradhan,15FARZAN}.
When models are certain or can be reliably estimated,  mixed integer nonlinear programming that involves discrete decision variables can be applied\cite{15Althaher,17Solanki,16Yao,15Paterakis}. In this case, discrete or integer decision variables often represent shiftable loads that can be adjusted to change demand curves.
If associated problem formulations do not involve discrete or integer decision variables, then conventional linear programming is applicable to residential demand management\cite{17Hayes}.

Convex optimization methods have been applied to solve a subproblem for demand response programs\cite{16MA,18Bahrami_PW}.
In\cite{16MA}, for example, an optimization problem that addressed a trade-off between  payments and the discomfort of users
was formulated, yielding both integer and continuous variables.
The problem was further separated into two subproblems, and one of them was solved using convex optimization methods.

To relax assumptions on entities in the system or probably to address system uncertainties, reinforcement learning for demand response  has been extensively studied recently\cite{VAZQUEZCANTELI20191072}.
Through the learning process, an agent or agents can learn how to optimize users' consumption patterns.
Considering the peak-to-average ratio,
Bahrami\emph{ et al.}\cite{18Bahrami} proposed an actor-critic structure that reduced the expected
cost of users in the aggregate load. Dehghanpour\emph{ et al.}\cite{18Dehghanpour} considered air conditioning loads as agents and
optimized their consumption patterns through modifying
the temperature set-points of the devices. Both consumption costs and users' comfort preferences were addressed.

Given the aforementioned state-of-the-art methodologies, a few points should be noted.
While shifting loads through demand response programs proves promising, it is worth mentioning that load shifting
may have impact on distribution systems\cite{16McKenna,1611Paterakis}.
In addition to routine operations, demand response programs can be utilized for emergency operations\cite{16AGHAEI}.
For a comprehensive review, the reader can refer to\cite{16HAIDER} and\cite{16ESTHER} on residential demand  and to\cite{16Safdarian} on a case study illustrating
the benefits of  demand response in consideration of residential users and utilities' load factor.

\subsection{Methods Without Load Factor}

This subsection examines recent studies on demand response programs without explicitly involving the concept of the load factor;
an energy management system has been a dominant way of realizing a demand response program in response to utility
pricing signals in the literature. Here is a list of examples. Hansen \emph{et al.}\cite{18Hansen} investigated a home energy management system that automated the energy usage. Observable Markov decision process approaches were proposed to minimize the household electricity bill.
Shafie-Khah \emph{et al.}\cite{18Shafie-Khah} proposed a stochastic energy management system that considered
uncertainties of the distributed renewable resources and the EV availability for charging and discharging.
Adika and Wang\cite{14Adika} examined a day-ahead demand-side bidding approach that maximized residential users' benefits.
Rastegar \emph{et al.}\cite{18Rastegar} constructed a two-level framework for residential energy management.
Customers minimized their payment costs and sent out the desired power scheduling of appliances at the first level;
 a multiobjective optimization framework was employed to improve technical characteristics of the distribution system at the second level.

Some authors focuses on the reduction of peak loads in residential demand response.
Vivekananthan \emph{et al.}\cite{14Vivekananthan} investigated
a reward based demand response algorithm that could shave network peaks.
Zhou \emph{et al.}\cite{18Zhou} established a multiobjective model of time-of-use and stepwise
power tariff for residential users, yielding load shifting from peak to
off-peak periods.

Most recently, learning based approaches to residential demand response have emerged, see, for example, \cite{19Ghasemkhani,15Wen}  and \cite{Ruelens17}.
In addition to direct investigations on methodologies, tools have been developed to facilitate existing demand response processes.
For instance, Paterakis \emph{et al.}\cite{16Paterakis} employed artificial neural networks and
the wavelet transform to predict the response of residential loads to price signals.
Wang \emph{et al.}\cite{18wang} developed  a method for estimating the residential demand response baseline.

Finally, it is worth mentioning that residential user behaviors may be studied in an aggregate manner, thereby introducing the concept of
residential demand aggregation\cite{19Elghitani}. Two  emerging topics are demand response methods for residential heating, ventilation, and air conditioning\cite{19Ma,17Erdinc,17Mahdavi}
and demand response methods for residential plug-in EVs\cite{18Rassaei,18Pal,18Munshi}.

\section{System Models}\label{sec_model}
This section discusses mathematical models describing the power consumption of home appliances, energy costs of residential users, and load factor of the grid. The day-ahead pricing scheme is considered. To realize an automatic demand response program, we consider an Internet of Things based environment with home appliances able to transmit and receive signals\cite{7004894,6935003}. These home appliances can be controlled by a central controller using IEEE standards, such as IEEE~802.3, IEEE~802.11, or IEEE~1901. Fig.~\ref{fig:framework} presents the system diagram and  Table~\ref{notations} summarizes notations and acronyms used throughout this paper.

\begin{figure}
\includegraphics[width=\linewidth]{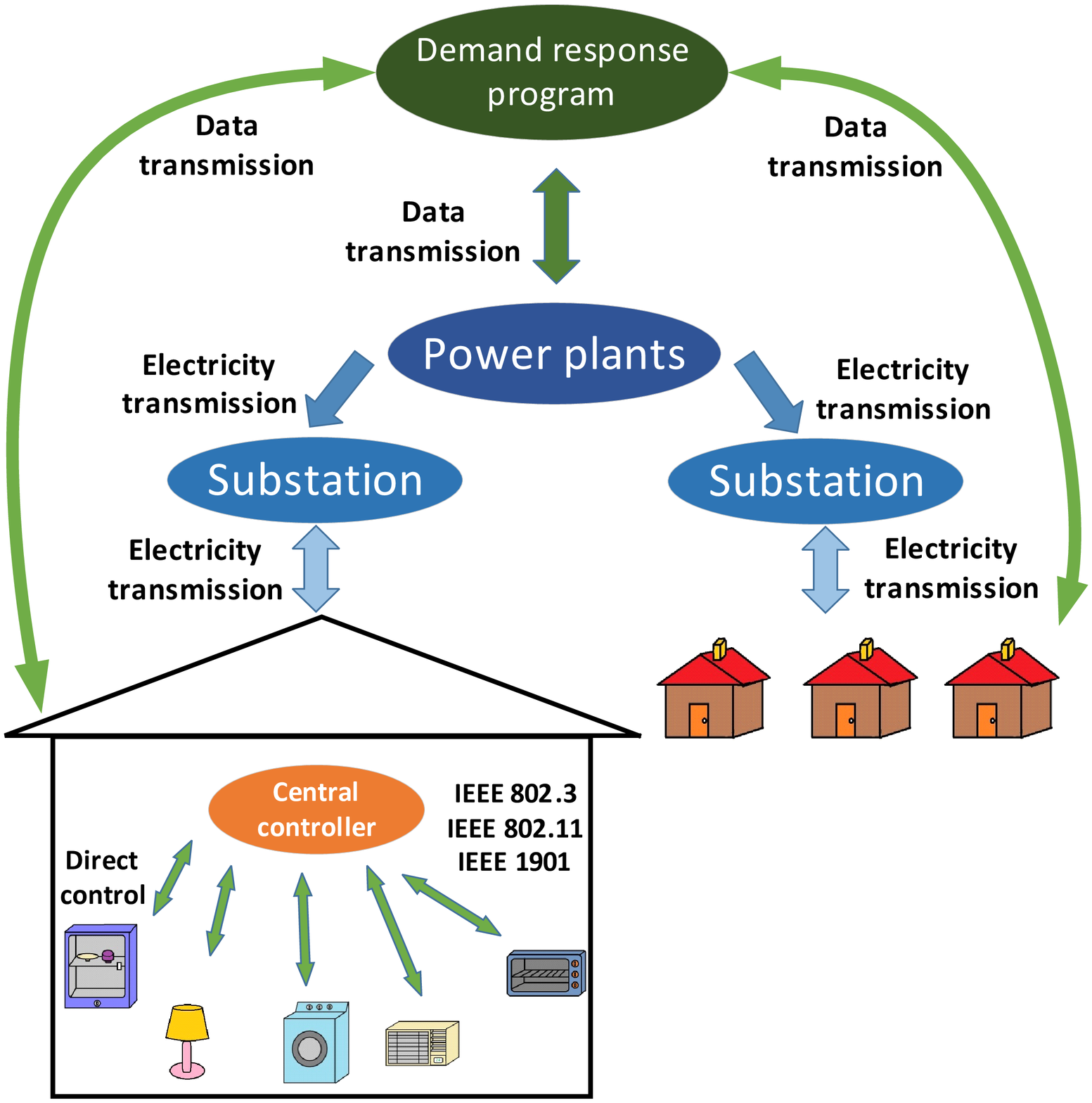}
\caption{System diagram of power suppliers and users in an Internet of Things based environment.}
\label {fig:framework}
\end{figure}

\subsection{Home Appliances}
Home appliances can be classified into three types~\cite{6461496}: unshiftable and inflexible appliances (the corresponding set is denoted by $\mathcal{A}$), shiftable and inflexible  appliances (the corresponding set is denoted by $\mathcal{A}_{S}$), and shiftable and flexible appliances (the corresponding set is denoted by $\mathcal{A}_{SF}$).
The loads induced by using those appliances can  be classified accordingly.
Let $h\in{\mathcal{H}}$ denote the time slot and $\Delta h$ denote the slot duration.

\begin{table}
\caption{Notations}
\centering
\label{notations}
\begin{tabular}{|P{1.3cm}|m{6.7cm}|}
\hline
$h$ & Time slot $h$ where $h\in \mathcal{H}$ \\ \hline
$H$ & Observation horizon (the size of $\mathcal{H}$) \\ \hline
$\mathcal{A}$& Set of unshiftable and inflexible appliances\\ \hline
$a$ & Unshiftable and inflexible appliance   \\ \hline
$p_{a}^h$ & Amount of power for appliance $a$ to operate in time slot $h$  \\ \hline
$\mathcal{A}_{S}$& Set of shiftable and inflexible appliances\\ \hline
$b$ & Shiftable and inflexible appliance    \\ \hline
$h_b$ & Working time slot of appliance $b$  \\ \hline
$s_b$ and $e_b$  & Start and end time slots of appliance $b$   \\  \hline
$w_b$ & Total working hours of appliance $b$  \\ \hline
$p_{b}^h(h_b)$ & Amount of power for appliance $b$ to operate in time slot $h$\\ \hline
$\mathcal{A}_{SF}$& Set of shiftable and flexible appliances\\ \hline
$c$ & Shiftable and flexible appliance  \\ \hline
 $s_c$ and $e_c$  & Start and end time slots of appliance  $c$ \\  \hline
$p_{c}^h$ & Amount of power for appliance $c$ to operate in time slot $h$\\ \hline
$p_c^{\max}$ & Maximum operating power of appliance $c$ \\ \hline
$p_c^{\min}$ & Minimum operating power of appliance $c$ \\ \hline
$p_{c,\text{agr}}^{\min}$ & Aggregate minimum power for using appliance $c$  \\ \hline
 $s_{\text{ev}}$ and $e_{\text{ev}}$  & Start and end time slots of EV battery charging\\  \hline
$p_{\text{ev}}^h$ & EV charging power in time slot $h$  \\ \hline
$p_{\text{ev}}^{\max}$ & Maximum charging rate of an EV \\ \hline
$B_{\text{ev}}^0$ & Initial capacity of the EV battery \\ \hline
$B_{\text{ev}}^{\min}$ &  Minimum capacity of the EV battery\\ \hline
$B_{\text{ev}}^{\max}$ &  Maximum capacity of the EV battery\\ \hline
$B^0$ & Initial capacity of the residential energy storage system \\ \hline
$B^h$ & Energy level of the residential energy storage system in time slot  $h$ \\ \hline
$B^{\min}$ &  Minimum capacity of the residential energy storage system\\ \hline
$B^{\max}$ &  Maximum capacity of the residential energy storage system\\ \hline
$u^h$  &  Charging and discharging control for the residential energy storage system in time slot  $h$\\ \hline
$r^h$  & Amount of power provided by renewable energy sources in time slot  $h$\\ \hline
\end{tabular}
\end{table}
\subsubsection{Unshiftable and Inflexible Appliances}
Some home appliances are used during specific time periods. For example, lights must be turned on in the evening, and it may not be possible to shift the time of use or adjust their switching time. Cooking appliances, such as electric pots, roasters, and microwave ovens, are also used in specific time slots. Residential users may use entertainment electronics such as televisions and computers only after work or school. These appliances are thus regarded as unshiftable and inflexible. For $a\in{\mathcal{A}}$, we denote $p_{a}^h$ as the amount of power for appliance $a$ to operate in time slot $h$.

\subsubsection{Shiftable and Inflexible Appliances}

Shiftable and inflexible  appliances are used at potentially any time, but their power consumption when performing a specified job is fixed (and thus inflexible). Vacuum cleaners and washing machines belong to this type. The time at which an automatic vacuum machine or a washing machine is operated may not be crucial (assuming that they do not make too much noise; otherwise, they should not be operated at night). A user could prescribe a possible time window, and then the demand response program could decide when the device operates according to these instructions. The time window must be sufficient for the appliance to complete the job.

Let $w_b$ denote the total working hours required by appliance~$b\in{\mathcal{A}_{S}}$.
Suppose that a residential user  sets
the acceptable start and end time slots of appliance $b$'s operation as $s_b$ and $e_b$, respectively ($s_b\leq e_b$).
Let ${C_{w_b}^{s_b \to e_b}}$ denote the set of all $w_b$-combinations of working time slots within the time window $[s_b,e_b]$.
A demand response program thus selects certain working time slots represented by $h_b$ from ${C_{w_b}^{s_b \to e_b}}$  for appliance $b$, i.e.,
\begin{equation}\label{eq_As}
h_b\in{C_{w_b}^{s_b \to e_b}} \qquad \forall{b\in{\mathcal{A}_{S}}}.
\end{equation}
 For instance, if $s_b=15$, $e_b=17$ , and $w_b$ = 2 (i.e., appliance $b$ needs 2 hours to finish its work), then
\begin{equation*}
C_{w_b}^{s_b \to e_b}=\left\{\left\{15, 16 \right\},\left\{15, 17 \right\},\left\{16, 17 \right\} \right\}
\end{equation*}
and $h_b$ could be equal to $\left\{15, 16 \right\}$, $\left\{15, 17 \right\}$ or $\left\{16, 17 \right\}$.
If $h_b=\left\{15, 16 \right\}$, then appliance $b$ operates from 14:00 to 16:00.
The associated power for appliance $b$ to operate in time slot $h$ is denoted
\begin{equation*}
 p_{b}^h(h_b)
\end{equation*}
 where $p_{b}^h(h_b)=0$ if $h_b \cap \left\{h \right\}=  \emptyset$.

\subsubsection{Shiftable and Flexible Appliances}
The use of shiftable and flexible appliances can be adjusted in two dimensions: when the appliances are used and how much power they should consume.
Heaters, air conditioners, and  EV charging stations can yield shiftable and flexible loads.
For example, residential users can lower an air conditioner's temperature setting (flexibility) and delay its use (shiftability) as they wish~\cite{6693775}.

Let $p_{c}^{h}$ denote the  amount of power for appliance $c$ to operate in time slot $h$, where $c \in \mathcal{A}_{SF}$, and
$s_c$ and $e_c$ denote the start and end time slots of appliance $c$, respectively.
The following constraints are imposed on $p_{c}^{h}$~\cite{5540263}:
\begin{equation}\label{eq_cst_c}
 p_c^{\min} \leq p_{c}^{h} \leq p_c^{\max}   \mbox{ and }   \sum_{h= s_c }^{e_c}   p_{c}^h \geq  p_{c,\text{agr}}^{\min}
\end{equation}
where  $p_c^{\min}$ and $p_c^{\max}$ represent  the minimum and maximum operating power, respectively, and
 $p_{c,\text{agr}}^{\min}$ represents the aggregate minimum power for using appliance $c$, which pertains to user comfort.
To ensure that feasible $p_{c}^{h}, h=s_{c},...,e_{c},$ exist, we assume
\begin{equation}\label{eq_ass_pc}
(e_c- s_c+1) p_c^{\max}  >  p_{c,\text{agr}}^{\min}.
\end{equation}

Although charging an EV can be considered as a shiftable and flexible load, we excluded EVs from $\mathcal{A}_{SF}$  because the charing power of an EV depends on the remaining capacity of the vehicle's battery, which imposes additional constraints~\cite{6674101}.
Let $p_{\text{ev}}^{h}$ denote the EV charging power  in time slot $h$, and $s_{\text{ev}}$ and $e_{\text{ev}}$ be the start and end time slots of the battery charging, respectively.
The following constraints on the charging rate and EV battery capacity  must be satisfied:
\begin{equation}\label{eq_cst_ev}
  0\leq p_{\text{ev}}^{h} \leq p_{\text{ev}}^{\max} \mbox{ and } B_{\text{ev}}^{\min} \leq  B_{\text{ev}}^0    + \sum_{h=s_{\text{ev}}}^{e_{\text{ev}}}  p_{\text{ev}}^{h} \Delta h  \leq   B_{\text{ev}}^{\max}
\end{equation}
where $p_{\text{ev}}^{\max}$ represents the maximum charging rate, $B_{\text{ev}}^{\min}$ is the minimum capacity of the EV battery, $B_{\text{ev}}^0$ is the initial capacity, and $B_{\text{ev}}^{\max}$ is the maximum capacity. To ensure that feasible $p_{\text{ev}}^{h}, h=s_{\text{ev}},...,e_{\text{ev}},$ exist, we assume
\begin{equation}\label{eq_ass_pev}
    B_{\text{ev}}^0    +  (e_{\text{ev}}- s_{\text{ev}}+1)  p_{\text{ev}}^{\max} \Delta h >  B_{\text{ev}}^{\min}.
\end{equation}

\subsection{Energy Cost and Load Factor}

A residential home can be integrated with renewable energy sources (RESs)
and equipped with a storage system for energy management. Let $B^h$ be the energy level of the storage system ($B^0$ represents the initial energy level), $r^h$ be the expected charging power from RESs, and $u^h$ be the control law that dictates the amount of power being charged to or discharged from the energy storage system.
The storage dynamics can be described as\cite{7927719}
 \begin{equation}\label{eq_storage}
   B^h=    B^{h-1} +  ( r^h-u^h) \Delta h.
 \end{equation}
 The constraints on the energy storage system are
  \begin{equation}\label{eq_cst_storage}
  0 \leq  B^h  \leq B^{\max}  \qquad  \forall h \in \mathcal{H}
 \end{equation}
 where $B^{\max}$ is the maximum storage capacity.

The total energy extracted from the grid in time slot $h$, denoted by $E_{\text{total}}^h$, can be expressed as
\begin{equation}
\label{alg:total}
\begin{aligned}
{E_{\text{total}}^h}= \max \{   \sum_{a\in{\mathcal{A}}}p_{a}^h+\sum_{b\in{\mathcal{A}_{S}}}p_{b}^h(h_b)\\
+\sum_{c\in{\mathcal{A}_{SF}}}p_{c}^h+p_{\text{ev}}^h   -u^h , 0\}  \Delta h.
\end{aligned}
\end{equation}
If $u^h$ is greater than the total power demand, then $E_{\text{total}}^h=0$ and the excess energy is discarded.
This situation can happen when $r^h$ is too large to be stored in the energy storage system and consumed by residential appliances.
If $u^h<0$, then $-u^h$ is the amount of power delivered to the energy storage system from the power grid.
The total energy cost can be obtained by multiplying the total energy consumption by the electricity price $\lambda^h$:
\begin{equation}
\sum_{h \in \mathcal{H} }E_{\text{total}}^h \lambda^h.
\label{alg:cost}
\end{equation}

Unlike customers, utilities are principally concerned with the load factor.
The load factor is critical because generation cost and grid quality have a direct connection with the load factor. A higher load factor  implies a more stable power grid and a lower cost of power generation\cite{14Gonen}. The load factor can be defined as the ratio of the average demand to the maximum  demand \cite{4652590,4914742}.
The following performance index can be used in a demand response program offered by utilities to improve their load factor:
\begin{equation}
 \frac{
 \underset{h\in \mathcal{H}}\sum E_{\text{total}}^h  /H  }{ \underset{h\in \mathcal{H}}{\max} \;E_{\text{total}}^h }
\label{alg:loadfactor}
\end{equation}
where $H$ is the size of $\mathcal{H}$ and represents the observation horizon.

\section{Pareto Optimal Demand Response Program}

This section proposes the PODR program that can be offered by utilities to residential users.
An MOP pertaining to power scheduling is formulated. The objectives of the optimization problem are to minimize the energy cost of a residential user in~(\ref{alg:cost}) and maximize the load factor in~(\ref{alg:loadfactor}).
To construct a stochastic search scheme for exploration and exploitation of the decision space,  we analyze the associated decision variables and propose a few methods of feasible value generations and mutation and crossover operations that preserve feasibility.
A multiobjective evolutionary algorithm for constructing the PODR program  is developed accordingly.
Finally, the algorithm complexity is discussed.

The MOP is formulated as
\begin{equation}\label{eq_MOP}
\begin{split}
  \underset{\bm{x}}{\min}  &\; f_{\text{cost}}(\bm{x})=\sum_{h \in \mathcal{H} }E_{\text{total}}^h  \lambda^h\\
    \underset{\bm{x}}{\max} & \; f_{\text{LF}}(\bm{x})= \frac{\underset{h\in \mathcal{H}}\sum E_{\text{total}}^h/H }{\underset{h\in \mathcal{H} }{\max}\;E_{\text{total}}^h}\\
 \text{subject to }   & (\ref{eq_cst_c}), (\ref{eq_cst_ev}), \text{ and } (\ref{eq_cst_storage})\qquad  \forall c \in \mathcal{A}_{SF}
\end{split}
\end{equation}
where $\bm{x}$ denotes the decision variable vector that contains decision variables $h_b$, $p_{c}^h$, $p_{\text{ev}}^h$, and $u^h$ for all $b,c,$ and $h$.
The objective functions are conflicting, so the global optimal solution that optimizes both objective functions simultaneously does not exist. To solve the MOP in~(\ref{eq_MOP}), Pareto optimality is adopted\cite{7050260}.
 A feasible point $\bm{x}'$, i.e., an $\bm{x}'$ that satisfies the constraints, dominates another feasible point $\bm{x}''$ if the conditions $f_{\text{cost}}(\bm{x}')\leq f_{\text{cost}}(\bm{x}'')$ and $f_{\text{LF}}(\bm{x}') \geq f_{\text{LF}}(\bm{x}'')$ hold true with at least one strict inequality.
A solution is nondominated (also called Pareto optimal or Pareto efficient) if improving one objective value must yield a degradation in the other objective value\cite{J13}.
A set of Pareto optimal solutions or nondominated solutions is desired.  A nondominated solution should be selected from the set on the basis of its associated performance represented by the Pareto front.

In the following discussions, we adopt a few terminologies used in genetic algorithms\cite{17Chong}.
Genetic algorithms are stochastic search techniques that have roots in genetics and employ a population based method to find solutions to optimization problems conventionally with a single objective.
These algorithms begin with
 a set of points randomly initialized. The set is termed the initial population. Points in the population are evaluated on the basis of the objective function, yielding function values called the fitness.
With the help of the fitness, a set of new points are generated using mutation and crossover operations.
These operations together with a selection operation based on the fitness are applied to improve an average fitness value from population to population.
The aforementioned procedure repeats iteratively to produce new populations until a stopping criterion is satisfied.

 In genetic algorithms, the mutation is performed on a candidate point and derives a new point termed a mutant. The probability of having a mutant is dictated by a mutation rate.
If the candidate point is represented by a binary string using an encoding scheme, then a typical mutation can be designed as stochastically complementing bits from 0 to 1 or vice versa.
While the mutation is applied to one candidate point, the crossover is performed on a pair of candidate points called the parents and produces a corresponding pair of points called the offspring.
The probability of performing the crossover operation is dictated by  a crossover rate.
In the case of using the binary representation, a typical crossover operation can be realized through an exchange of substrings of the parents.

To address constraints of optimization problems in genetic algorithms, penalty methods can be used. A penalty function is added to the objective function for point evaluation.
The fitness of a point becomes a sum of the objective function value and a penalty function value.
For an infeasible point, i.e., violating the constraints,
its penalty function value is nonzero  (negative for maximization problems and positive for minimization problems) and thus penalizes the objective value.
Because points with better fitness are prone to be kept in populations during the algorithm iterations, infeasible points can be gradually removed, leading to a feasible set.

Genetic algorithms provide a generic framework for solving optimization problems.
Suitable modifications can render them more efficient and powerful.
For example, we may use certain schemes dedicated to the problem of interest that randomly generate feasible points
 and provide mutation and crossover operations that preserve the feasibility of those points.
As such, algorithm efficiency can be improved because more computations are spent on improving feasible points instead of finding feasible ones.
Furthermore, by incorporating the concept of Pareto optimality into the fitness evaluation, we may design algorithms that can address multiple objectives, which are termed
multiobjective evolutionary algorithms.

Although a few multiobjective evolutionary algorithms are available for finding solutions of MOPs, most of them consider either pure continuous decision variables or pure discrete decision variables.
The situation in which
both continuous and discrete decision variables are involved, as in our scenario,
has not been well addressed\cite{MOEA_bk1}.
In addition, the constraints in~(\ref{eq_MOP}) can be difficult to address using conventional constraint handling techniques.
Methods of  feasible value generations and mutation and crossover operations that preserve feasibility are required.

To solve~(\ref{eq_MOP}), we first consider an encoding scheme for discrete variables and present the associated mutation and crossover operations.
Continuous variables are then addressed.
For discrete variable $h_b$, a feasible value can be readily generated according to~(\ref{eq_As}).
The following encoding scheme and associated mutation and crossover operations are adopted.
Let $\phi(h_b)$ be a binary representation of $h_b$ with length  $e_b-s_b+1$ and $[\phi(h_b)]_j$ denote the $j$th bit of $\phi(h_b)$.
The encoding scheme $\phi(\cdot)$ is designed as
\begin{equation}\label{eq_encoder}
  [\phi(h_b)]_j=
  \left\{
    \begin{array}{ll}
      1, & \hbox{ if } j+s_b-1 \in  h_b\\
      0, & \hbox{otherwise}
    \end{array}
  \right.
\end{equation}
for $j=1,2,...,e_b-s_b+1$. For instance, if
$s_b=15,e_b=17,w_b = 2$, and  $h_b=\left\{15, 16 \right\}$, then $\phi(h_b)=\; $110; if $h_b=\left\{15, 17 \right\}$, then $\phi(h_b)=\; $101.
For the mutation operation, we randomly find a pair   $([\phi(h_b)]_{j},[\phi(h_b)]_{j'})=(0,1)$ and switch their values.
For the crossover operation, we apply logic operation ``OR''  to two binary representations,
and then randomly choose $w_b$ bits with value 1 from the offspring and  convert the other bits with value 1 to value 0.
The condition in~(\ref{eq_As}) is thus satisfied  through the use of  the mutation and crossover operations.

For continuous decision variables $p_c^{h},p_{\text{ev}}^{h},$ and $u^h$, we note that the following mutation and crossover operations can produce feasible points in the decision space if  points involved are all feasible.

\begin{thm}\label{thm_gen_opr}
Let $x_{\text{new}}^{h}$ denote a mutant or offspring.
The mutation or crossover is performed according to
\begin{equation}\label{eq_conti_opr}
x_{\text{new}}^{h}= \delta x^{h}  +  (1-\delta) {x'}^{h}
\end{equation}
where $\delta \in [0,1]$ is a random number for all $h$.
For the mutation operation, $x^{h}$ is a  point in the population and ${x'}^{h}$ is a point that is randomly generated.
For the crossover operation, $x^{h}$ and ${x'}^{h}$ are distinct points in the population.
If $ x^{h}$ and ${x'}^{h}$ in~(\ref{eq_conti_opr}) are feasible, then  the mutant or offspring $x_{\text{new}}^{h}$ is feasible.
\end{thm}
\begin{proof}
 Consider   $x^{h}=p_{c}^{h}$. If   $ p_c^{h}$ and ${p'}_c^{h}$ are feasible, then
\begin{equation*}
\begin{split}
   & p_{c}^{h}\in [ p_c^{\min} ,   p_c^{\max}] , {p'}_{c}^{h}\in [ p_c^{\min} ,  p_c^{\max}] ,\sum_{h= s_c }^{e_c}   p_{c}^h \geq  p_{c,\text{agr}}^{\min}  ,  \mbox{ and } \\
    & \sum_{h= s_c }^{e_c}  {p'}_{c}^h \geq  p_{c,\text{agr}}^{\min}.
\end{split}
\end{equation*}
We have
\begin{equation*}
\begin{split}
   & \delta p_c^{h}  +  (1-\delta) {p'}_c^{h} \in [ p_c^{\min} ,   p_c^{\max}]  \mbox{ and }     \\
    &  \sum_{h= s_c }^{e_c} \delta p_c^{h}  +  (1-\delta) {p'}_c^{h} \geq  p_{c,\text{agr}}^{\min}
\end{split}
\end{equation*}
for $\delta\in [0,1]$. Therefore, $p_{\text{new}}^{h}= \delta p_c^{h}  +  (1-\delta) {p'}_c^{h}$ is feasible.
A similar argument can be performed to show the feasibility property when $x^{h}=p_{\text{ev}}^{h}$.

Consider $x^{h}=u^{h}$.
 If   $ u^{h}$ and ${u'}^{h}$ are feasible, then we have
\begin{equation}\label{eq_feasible_u}
\begin{split}
   &  B^{h-1} +  ( r^h-u^h) \Delta h  \in  [0,B^{\max}]  \mbox{ and } \\
    & {B'}^{h-1} +  ( r^h-{u'}^h) \Delta h  \in  [0,B^{\max}] .
\end{split}
\end{equation}
Define
\begin{equation*}
  B_{\text{new}}^{h}= \delta B^{h} +  (1-\delta) {B'}^{h}     \mbox{ and }  u_{\text{new}}^{h}= \delta u^{h}  +  (1-\delta) {u'}^{h}.
\end{equation*}
By~(\ref{eq_feasible_u}), we have
\begin{equation*}
  B_{\text{new}}^{h-1} +  ( r^h-u_{\text{new}}^{h}) \Delta h \in  [0,B^{\max}].
\end{equation*}
Note that
\begin{equation*}
{\small
  \begin{split}
         &    B_{\text{new}}^{h}\\
{ }={ }  & \delta (B^{h-1}+( r^h-u^h) \Delta h) +  (1-\delta)( {B'}^{h-1} +( r^h-{u'}^h) \Delta h)         \\
   { }={ }  & B_{\text{new}}^{h-1} +  ( r^h-u_{\text{new}}^{h}) \Delta h .
\end{split}
}
 \end{equation*}
Therefore,  $B_{\text{new}}^{h} \in  [0,B^{\max}]$, which implies that $ u_{\text{new}}^{h}$ is feasible. \hfill $\Box$
\end{proof}

If there are methods of randomly generating feasible points, then the feasibility can be preserved during
the solution process using the mutation and crossover operations in Theorem~\ref{thm_gen_opr}.
Algorithms~\ref{alg_pch},~\ref{alg_EV}, and~\ref{alg_uh} present those methods, and their effectiveness are further confirmed by the following theorem.

\begin{algorithm}
\caption{Generation of Feasible $p_c^{h}$}
\label{alg_pch}
\begin{algorithmic}
\Require $s_c,e_c, p_c^{\min} ,  p_c^{\max} $, and $ p_{c,\text{agr}}^{\min}$.
\Ensure Feasible $p_c^{h}, h=s_c,...,e_c,$  that satisfy~(\ref{eq_cst_c}).
\State
Generate $p_c^{h}$ randomly from $[p_c^{\min} ,p_c^{\max} ]$.\\
\textbf{while} $ \sum_{h= s_c }^{e_c}   p_{c}^h <  p_{c,\text{agr}}^{\min} $ \textbf{do}
    \begin{equation}\label{eq_pc_increase}
  p_c^{h}:= p_c^{h}+ (p_c^{\max}-p_c^{h}) \delta^h
    \end{equation}
where  $\delta^h \in [0,1]$ is a random number.\\
\textbf{end while}
\end{algorithmic}
\end{algorithm}

\begin{algorithm}
\caption{Generation of Feasible $p_{\text{ev}}^{h}$}
\label{alg_EV}
\begin{algorithmic}
\Require $s_{\text{ev}},e_{\text{ev}}, p_{\text{ev}}^{\max} ,  B_{\text{ev}}^0,B_{\text{ev}}^{\min}$, and $B_{\text{ev}}^{\max}$.
\Ensure Feasible $p_{\text{ev}}^{h}, h=s_{\text{ev}},...,e_{\text{ev}},$  that satisfy~(\ref{eq_cst_ev}).
\State
Generate $p_{\text{ev}}^{h}$ randomly from $[0,p_{\text{ev}}^{\max} ]$.\\
\textbf{while} $\sum_{h=s_{\text{ev}}}^{e_{\text{ev}}}  p_{\text{ev}}^{h} \Delta h < B_{\text{ev}}^{\min} -  B_{\text{ev}}^0  $ \textbf{do}
    \begin{equation}\label{eq_pev_incease}
  p_{\text{ev}}^{h}:= p_{\text{ev}}^{h} + ( p_{\text{ev}}^{\max} -p_{\text{ev}}^{h})\delta^h
    \end{equation}
where  $\delta^h \in [0,1]$ is a random number.\\
\textbf{end while}
\State \textbf{if}  $\sum_{h=s_{\text{ev}}}^{e_{\text{ev}}}  p_{\text{ev}}^{h} \Delta h >  B_{\text{ev}}^{\max} -  B_{\text{ev}}^0 $  \textbf{then} \\
    \begin{equation}\label{eq_pev_trim}
  p_{\text{ev}}^{h}:=
        \frac{  B_{\text{ev}}^{\max} -  B_{\text{ev}}^0 }{\sum_{h=s_{\text{ev}}}^{e_{\text{ev}}}  p_{\text{ev}}^{h}\Delta h}    p_{\text{ev}}^{h}
    \end{equation}
\textbf{end if}
\end{algorithmic}
\end{algorithm}

\begin{algorithm}
\caption{Generation of Feasible $u^{h}$}
\label{alg_uh}
\begin{algorithmic}
\Require $B^0$ and $ r^h$.
\Ensure Feasible $u^{h}, h\in \mathcal{H}$  (i.e., the constraints in~(\ref{eq_cst_storage}) are satisfied).
\State
\textbf{for} $h=1:H $ \textbf{do}\\
Randomly generate $u^h$ such that
    \begin{equation}\label{eq_uh_recurence}
 u^h \in [  \frac{B^{h-1}+r^h\Delta h -B^{\max}}{\Delta h}  ,   \frac{B^{h-1}+r^h\Delta h }{\Delta h}].
    \end{equation}
\State
Evaluate
\begin{equation*}
B^{h}=B^{h-1}+  (r^h-u^h)\Delta h.
\end{equation*}
\textbf{end for}
\end{algorithmic}
\end{algorithm}

\begin{thm}\label{thm_feasibility}
  Under the assumptions described in~(\ref{eq_ass_pc}) and~(\ref{eq_ass_pev}), Algorithms~\ref{alg_pch},~\ref{alg_EV}, and~\ref{alg_uh} produce feasible values for $p_{c}^{h}, p_{\text{ev}}^{h},$ and $u^h$.
\end{thm}
\begin{proof}
According to~(\ref{eq_pc_increase}) and the fact that $p_c^h,h=s_c,...,e_c,$ are initially generated from $[p_c^{\min},p_c^{\max}]$,  variables $p_c^h,h=s_c,...,e_c$, with $p_c^h\geq p_c^{\min}$  approach $p_c^{\max}$ from the left as the number of iterations increases.
Owing to the assumption in~(\ref{eq_ass_pc}), the conditions in~(\ref{eq_cst_c}) hold true eventually when the values of $p_c^h,h=s_c,...,e_c,$ increase.
By a similar argument and according to~(\ref{eq_ass_pev}) and~(\ref{eq_pev_incease}), we have $ B_{\text{ev}}^{\min} \leq  B_{\text{ev}}^0    + \sum_{h=s_{\text{ev}}}^{e_{\text{ev}}}  p_{\text{ev}}^{h} \Delta h$ eventually
when the values of $ p_{\text{ev}}^{h},h=s_{\text{ev}},...,e_{\text{ev}},$ increase; however, it is possible that
an increment is too large to have $ B_{\text{ev}}^0    + \sum_{h=s_{\text{ev}}}^{e_{\text{ev}}}  p_{\text{ev}}^{h} \Delta h  \leq  B_{\text{ev}}^{\max}$.
To remedy this, we use~(\ref{eq_pev_trim}) to reduce  the values of $p_{\text{ev}}^{h},h=s_{\text{ev}},...,e_{\text{ev}}$. When~(\ref{eq_pev_trim}) is executed, we have the following two results:
$p_{\text{ev}}^{h}$ on the left-hand side of~(\ref{eq_pev_trim}) satisfies
$ p_{\text{ev}}^{h} < p_{\text{ev}}^{\max}$ because the term
$ ( B_{\text{ev}}^{\max} -  B_{\text{ev}}^0 )/ \sum_{h=s_{\text{ev}}}^{e_{\text{ev}}}  p_{\text{ev}}^{h}\Delta h$ on the right-hand side of~(\ref{eq_pev_trim}) is less than 1; and
$p_{\text{ev}}^{h}$ on the left-hand side of~(\ref{eq_pev_trim}) satisfies
  $  \sum_{h=s_{\text{ev}}}^{e_{\text{ev}}}  p_{\text{ev}}^{h}=( B_{\text{ev}}^{\max} -  B_{\text{ev}}^0 )/\Delta h.$
The conditions in~(\ref{eq_cst_ev}) are then satisfied.
Finally, we note that~(\ref{eq_uh_recurence}) implies~(\ref{eq_cst_storage}), illustrating the feasibility of  $u^h$.
\hfill $\Box$
\end{proof}

With the help of Algorithms~\ref{alg_pch}--\ref{alg_uh},
Algorithm~\ref{alg_PODR} presents the PODR program that can be offered by utilities to residential users.

\begin{algorithm}
\caption{Pareto Optimal Power Scheduling}
\label{alg_PODR}
\begin{algorithmic}
\Require Electricity price $\lambda^h$;  MOP in~(\ref{eq_MOP}) with
parameters $p_{a}^h$, $p_{b}^h$, $w_b$, $s_b$, $e_b$, $p_{c}^{\max}$, $p_c^{\min}$, $p_{c,\text{agr}}^{\min}$, $s_c$,
  $e_c$, $p_{\text{ev}}^{\max}$, $s_{\text{ev}}$, and  $e_{\text{ev}}$; mutation rate $\mu$;
nominal population size $N_{\text{nom}}$; maximum population size $N_{\max}$; and maximum iteration number  $t_{\max}$.
\Ensure PODR program.
\State  \textbf{Step 1)} Initialize $X(0)$: randomly generate working hours
$h_b$ for shiftable but inflexible appliances; apply Algorithms~\ref{alg_pch},~\ref{alg_EV}, and~\ref{alg_uh} to generate feasible values for $p_c^h$, $p_{\text{ev}}^h$, and $u^h$. Remove dominated points from $X(0)$.\\
\textbf{Step 2)} Let $t_{c}=0$.\\
\textbf{while }$t_c \leq t_{\max}$ \textbf{ do}\\
 \begin{description}
   \item[] \textbf{Step 2.1)} Clone points in $X(t_c)$ with clone rate $N_{\max}/N_{\text{nom}}$.
Apply mutation operation with  rate $\mu$ and crossover operation with rate $1-\mu$ defined in Theorem~\ref{thm_gen_opr} to cloned points.
Store the mutants and offspring in $X(t_c)$.
\item[] \textbf{Step 2.2)} Remove dominated points from $X(t_c)$. If $\mid X(t_c) \mid > N_{\text{nom}}$, then use an archive update method to reduce its size.
\item[]\textbf{Step 2.3)} Let $X(t_c+1)=X(t_c)$ and
 $t_c=t_c+1$.
\end{description}\\
\textbf{end while}\\
 \textbf{Step 3)} Select the knee of the approximate Pareto front associated with the approximate Pareto  set $X(t_{\max})$.
\end{algorithmic}
\end{algorithm}

In Algorithm~\ref{alg_PODR}, information about appliances and acceptable times of use is set first. Mutation and crossover operations are then performed over $X(t_c)$ in Step~2.1.  In Step~2.2, the archive is updated by removing some nondominated points if the population size $\mid X(t_c) \mid$ is too large. After a number of iterations, an approximate Pareto set and front are obtained in Step~3.
In practice, the knee of the approximate Pareto front is often preferred for several reasons: it can achieve excellent overall system performance if the front is bent;
it represents the solution closest to the ideal one that is not reachable; and it has rich geometrical and physical meanings\cite{15Zhang,09Rachmawati}.
 In Step~3, the knee solution is selected according to~\cite{7465803}:
\begin{equation}\label{alg:mmd}
\begin{split}
\bm{x}^*= & \arg \min_{ \bm{x}\in X(t_{\max}) } \; \frac{f_{\text{cost}}( \bm{x} )-\underset{ \bm{x}' \in X(t_{\max})}{\min} f_{\text{cost}}(\bm{x}')}{L_{\text{cost}}}\\
    & +\frac{\underset{ \bm{x}' \in X(t_{\max})}{\max} f_{\text{LF}}( \bm{x}') - f_{\text{LF}}(\bm{x})  }{L_{\text{LF}}}
\end{split}
\end{equation}
where
\begin{equation}
\begin{split}
L_{\text{cost}}{ }={ } & \underset{ \bm{x} \in X(t_{\max})}{\max} f_{\text{cost}}( \bm{x})-\underset{ \bm{x} \in X(t_{\max})}{\min} f_{\text{cost}}( \bm{x}) \mbox{ and }\\
L_{\text{LF}}{ }={ } &  \underset{ \bm{x}\in X(t_{\max})}{\max} f_{\text{LF}}( \bm{x})-\underset{ \bm{x} \in X(t_{\max})}{\min} f_{\text{LF}}( \bm{x})
\end{split}
\end{equation}
are the maximum spreads of the approximate Pareto front in the first and second dimensions, respectively, and
\begin{equation*}
\left[
  \begin{array}{cc}
     \frac{\underset{ \bm{x} \in X(t_{\max})}{\min} f_{\text{cost}}(\bm{x})}{L_{\text{cost}}} & \frac{\underset{ \bm{x} \in X(t_{\max})}{\max} f_{\text{LF}}( \bm{x})   }{L_{\text{LF}}} \\
  \end{array}
\right]^T
\end{equation*}
is the ideal vector.

As shown in~(\ref{alg:mmd}), the final Pareto optimal solution is selected using the system's knowledge of the maximum spreads, critical information in multicriteria decision making. Such information is not obtained using conventional single-objective optimization approaches.
The knee solution $\bm{x}^*$ corresponds to a Pareto optimal power consumption profile, yielding the PODR program.

The complexity of Algorithm~\ref{alg_PODR} mainly depends on the use of Pareto concepts that are the most computationally expensive.
The complexity can be roughly expressed as\cite{MOEA_bk1}
\begin{equation*}
O(2 t_{\max} N_{\text{nom}}(t_{\text{com}} +   N_{\text{nom}}-1    )   )
\end{equation*}
where $t_{\text{com}}$ represents the average computational time for objective evaluation.

\section{Simulation Results}
This section presents an examination of the power scheduling of 400 residential users.\footnote{In~\cite{4914742}, the total available capacity of a region was set at 2296 kW. In our scenario,
each residential user consumes 5.5 kW at most, which thus accounts for approximately 400 residential users in the region.}
Each residence was equipped with at most thirteen unshiftable and inflexible appliances, three shiftable and inflexible appliances, two shiftable and flexible appliances, one EV, and possibly
one energy storage system integrated with solar energy sources.
The exact number of appliances and associated types for a residential user were randomly chosen.
Solar energy data in\cite{EIA} were used.
Let ${\mathcal{H}={\left\{ 1,2,...,24 \right\}}}$ and
$\Delta h=1$.
 Table~\ref{parameter1} presents the associated settings, and the notation ``unif($h_1,h_2,t$)'' therein indicates that an appliance required $t$ hours of working time starting from a number randomly chosen from $\{h_1, h_1+1, ..., h_2 \}$. For example, appliance $a_4$ needed 5 hours to complete the task, and the start time slot could have been 16, 17, or 18.
Most parameter values were chosen according to\cite{Web1} and\cite{5991951}.

The day-ahead prices illustrated in Fig.~\ref{fig_price} were used~\cite{Web2}.
After analyzing total energy consumption of each month in 2017, we discovered that energy consumption peaked in July.
Four representative days in July were selected and listed in a decreasing order in terms of their price range as follows:
27 July 2017 (maximum price range),  6 July 2017 (25th percentile), 23 July 2017 (75th percentile), and 29 July 2017 (minimum price range).
Among them, 27 July 2017 also had the highest price.

 \begin{table}
\centering
\caption{Simulation Parameters}
\label{parameter1}
\begin{tabular}{|P{1.5cm}|P{2cm}|P{3cm}|}
\hline

                 Notations & Power  &Working Schedule\\ \hline
$p_{a_1}^h$                  & 0.02 kW    & $h \in \{17,18,...,24\}$               \\ \hline
$p_{a_2}^h$                  & 0.22 kW  & unif(18, 22, 3)                 \\ \hline
$p_{a_3}^h$                  & 0.2  kW  & unif(11, 13, 3)              \\ \hline
$p_{a_4}^h$                  & 0.2  kW  & unif(16, 18, 5)            \\ \hline
$p_{a_5}^h$                  & 0.7  kW  & unif(18, 22, 1)            \\ \hline
$p_{a_6}^h$                  & 1.3  kW  & unif(14, 16, 1)            \\ \hline
$p_{a_7}^h$                  & 0.2  kW  & unif(18, 22, 1)            \\ \hline
$p_{a_8}^h$                  & 0.08 kW  & unif(18, 20, 3)            \\ \hline
$p_{a_9}^h$                  & 0.05 kW  & $h \in \{1,2,...,24\}$             \\ \hline
$p_{a_{10}}^h$                  & 1.5 kW   & $h =8$             \\ \hline
$p_{a_{11}}^h$                 & 1.6 kW   & $h \in \{17,18\}$              \\ \hline
$p_{a_{12}}^h$                 & 0.2 kW   & $h \in \{1,2,...,24\}$            \\ \hline
$p_{a_{13}}^h$                 & 0.8 kW   & $h =17$            \\ \hline
\hline
Notations & \multicolumn{2}{c|}{Values}  \\ \hline
$s_{b_1},e_{b_1}$         & \multicolumn{2}{c|}{Random number from  $\{ 10, 11, 12, 13 \}$, $e_{b_1}=s_{b_1}+7$} \\ \hline
$p_{b_1}^h, w_{b_1} $       & \multicolumn{2}{c|}{1 kW if     $h_{b_1} \cap \left\{h \right\}\not=  \emptyset, w_{b_1}=1 $  }               \\ \hline
$s_{b_2},e_{b_2}$         & \multicolumn{2}{c|}{Random number from  $\{ 12, 13, 14, 15 \}$, $e_{b_2}=s_{b_2}+4$} \\ \hline
$p_{b_2}^h, w_{b_2} $       & \multicolumn{2}{c|}{1 kW if     $h_{b_2} \cap \left\{h \right\}\not=  \emptyset, w_{b_2}=2 $  }               \\ \hline
$s_{b_3},e_{b_3}$         & \multicolumn{2}{c|}{Random number from  $\{ 13, 14, 15, 16 \}$, $e_{b_3}=s_{b_3}+7$} \\ \hline
$p_{b_3}^h, w_{b_3} $       & \multicolumn{2}{c|}{2 kW if     $h_{b_3} \cap \left\{h \right\}\not=  \emptyset, w_{b_3}=2 $  }               \\ \hline
$s_{c_1}$         &\multicolumn{2}{c|}{12} \\ \hline
$e_{c_1}$         &\multicolumn{2}{c|}{24} \\ \hline
$s_{c_2}$         &\multicolumn{2}{c|}{Random number from  $\{ 20, 21, 22, 23  \}$} \\ \hline
$e_{c_2}$         &\multicolumn{2}{c|}{$s_{c_2}+9$}\\ \hline
$s_{\text{ev}}$         &  \multicolumn{2}{c|}{Random number from  $\{ 18, 19,..., 22  \}$} \\ \hline
$e_{\text{ev}}$         &\multicolumn{2}{c|}{$s_{\text{ev}}+11$}\\ \hline
$p_{c_1}^{\max},p_{c_2}^{\max}$                 & \multicolumn{2}{c|}{3 kW }            \\ \hline
$p_{c_1}^{\min},p_{c_2}^{\min}$                 & \multicolumn{2}{c|}{ 0.5 kW }            \\ \hline
$p_{c_1,\text{agr}}^{\min}$                 & \multicolumn{2}{c|}{ 29 kW }            \\ \hline
$p_{c_2,\text{agr}}^{\min}$                 & \multicolumn{2}{c|}{ 12 kW }            \\ \hline
$p_{\text{ev}}^{\max}$                 & \multicolumn{2}{c|}{3 kW }          \\ \hline
$B^{\max},B^0$                 & \multicolumn{2}{c|}{4 kWh, 1 kWh  }          \\ \hline
$B_{\text{ev}}^{\max}, B_{\text{ev}}^{\min}$    & \multicolumn{2}{c|}{24 kWh, $B_{\text{ev}}^{\min}=0.8 B_{\text{ev}}^{\max}$    }      \\ \hline
$B_{\text{ev}}^0$                 &  \multicolumn{2}{c|}{ Random number from  $[0.3B_{\text{ev}}^{\max},0.6B_{\text{ev}}^{\max}]$ }          \\ \hline
\end{tabular}
\end{table}

\begin{figure}
\includegraphics[width=\linewidth]{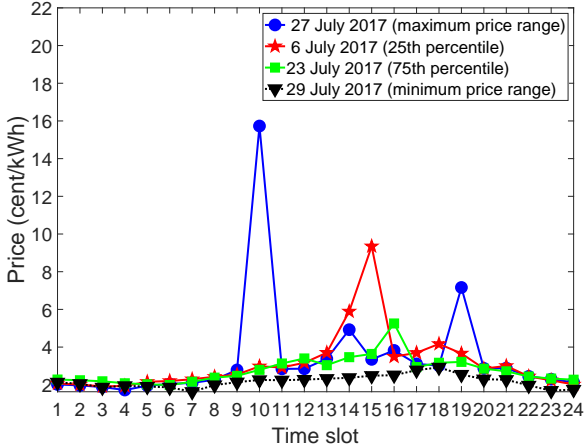}
\caption{Pricing schemes used in our simulations.}
\label{fig_price}
\end{figure}

\begin{figure}

\includegraphics[width=\linewidth]{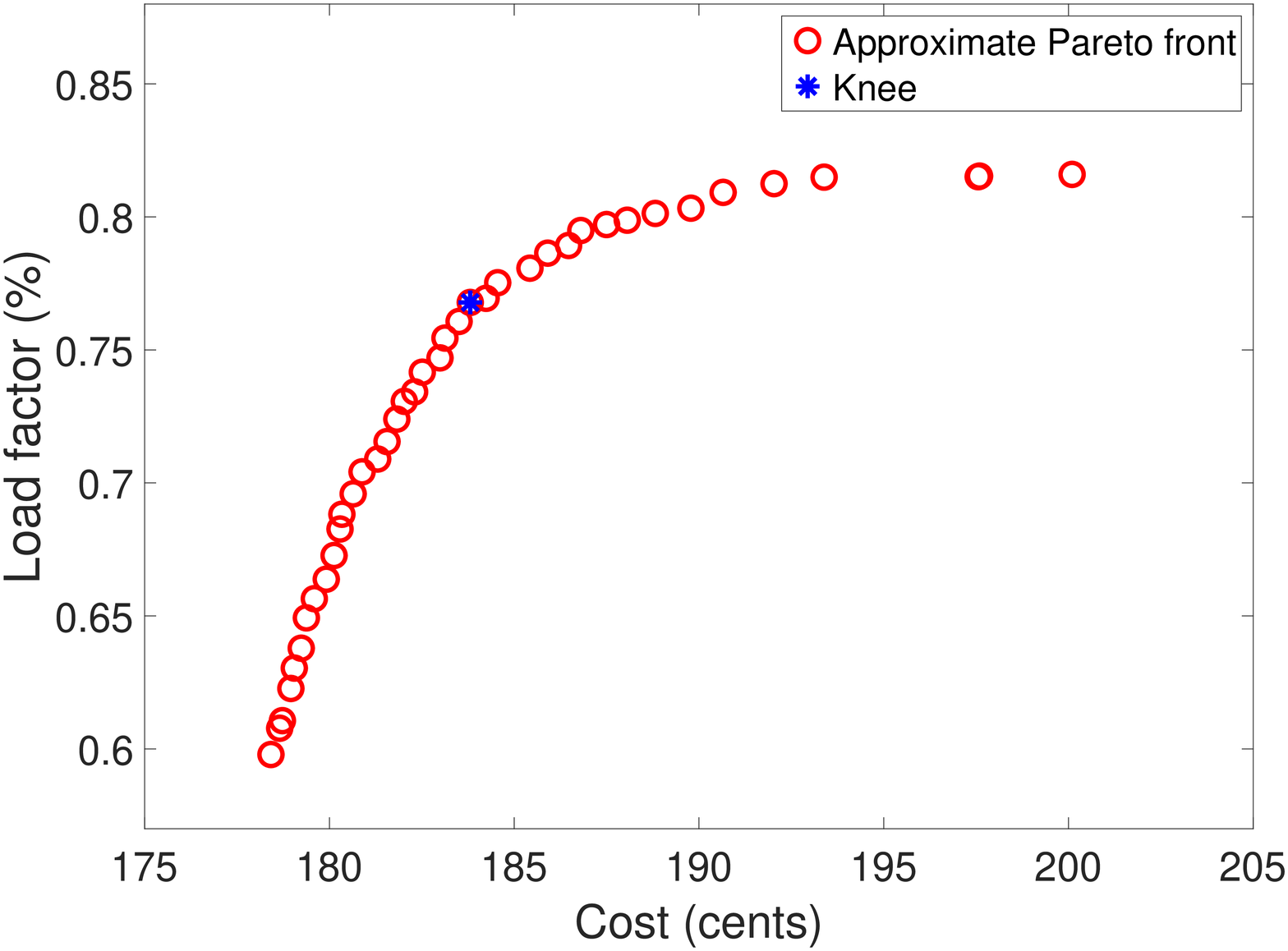}
\caption{Sampled approximate Pareto front obtained by solving~(\ref{eq_MOP}).}
\label{fig:pareto1}
\end{figure}

The PODR program for users was obtained using Algorithm~\ref{alg_PODR} with $\mu=0.8, N_{\text{nom}}=40,N_{\max}=400$ and $t_{\max}=400$.
Approximate Pareto fronts associated with residential users were produced.
Fig.~\ref{fig:pareto1} plots a sampled approximate Pareto front. Each point on the front corresponded to two values:  energy cost and load factor. The shape of the front confirmed that minimizing the cost and maximizing the load factor were  conflicting objectives.
The knee was selected according to~(\ref{alg:mmd}).

Our multiobjective approach was compared with an area-load method modified from\cite{13Kunwar}, a payment minimization method modified from\cite{4914742}, and load factor maximization and load variance minimization methods modified from\cite{11Sortomme} (see Appendix).  Owing to  constraint structures and different types of variables that were involved,
conventional stochastic search methods could hardly produce feasible points. Algorithms \ref{alg_pch}--\ref{alg_uh} were thus applied to produce a portion of initial points that were
used through the solution processes associated with the comparable methods.
Fig.~\ref{fig:allcompare} presents individual performance on the representative days.
Table~\ref{comparison_cost} summarizes the performance (expressed as a percentage) of each method.
 For the energy cost, the percentages were evaluated with respect to our method. A smaller percentage indicated a lower cost.
For the load factor, the percentages were evaluated with respect to the load variance minimization method. A larger percentage indicated a higher load factor.

Among those methods, the PODR program balanced the two conflicting objectives in an advantageous manner. This should be generally the case because our approach
 considered join optimization of the two objectives,  had robust constraint handling techniques presented in Algorithms~\ref{alg_pch}-\ref{alg_uh},
and employed dedicate mutation and crossover operations to preserve solution feasibility, as shown in Theorem~\ref{thm_feasibility}.
The modified area-load method also adopted a multiobjective approach; as compared to the PODR, it yielded a small improvement in the load factor by approximately 2.8\% (-7.8\%+10.6\%) but a large degradation in the energy costs by 14.5\%.
For single-objective optimization methods,
load variance minimization and load factor maximization yielded larger load factors than other methods because performance metrics related to the load factor were optimized; however, the resulting energy costs
for residential users, which were not considered during the solution process, were the worst among other methods. Finally, the modified payment minimization method yielded satisfactory costs for residential users but the worst load factor; the associated energy cost was higher than that of the PODR because that method  lacked suitable constraint handling techniques.

\begin{figure}
\begin{equation*}
  \begin{array}{c}
     \includegraphics[width=\linewidth]{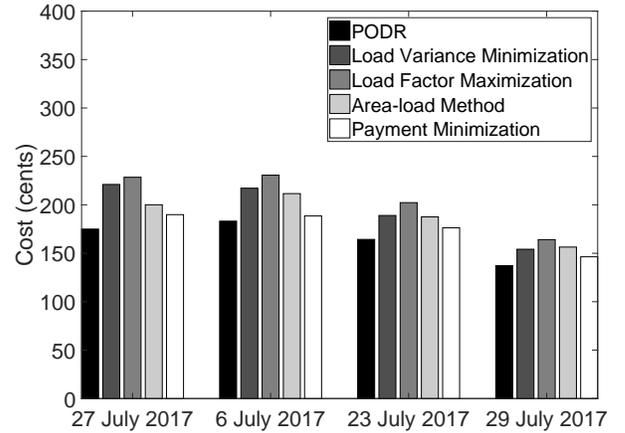}\\
    \mbox{(a)}\\
    \includegraphics[width=\linewidth]{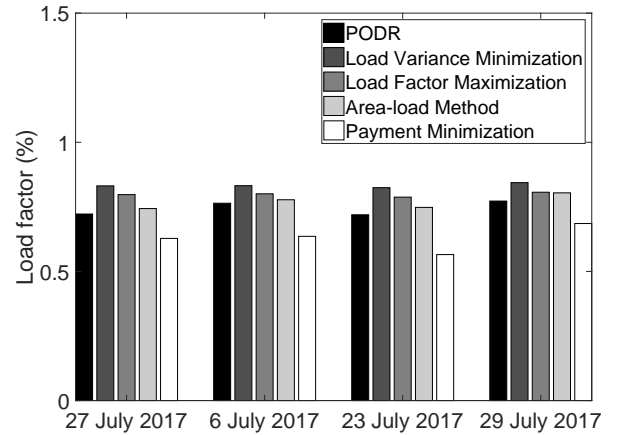} \\
     \mbox{(b)}
   \end{array}
\end{equation*}
\caption{(a) Energy costs and (b) load factors of 400 residential users.}
\label{fig:allcompare}
\end{figure}

\begin{table}
\caption{Comparison of Various Power Scheduling Methods}
\label{comparison_cost}
\centering
{\scriptsize
%\hspace{-0.5cm}
\begin{tabular}{|@{}c@{}|@{}c@{}|@{}c@{}|@{}c@{}|@{}c@{}|}
\hline
\multicolumn{5}{|c|}{ \footnotesize Cost Increase With Respect to }\\
\multicolumn{5}{|c|}{ \footnotesize  PODR Program}\\ \hline
\multirow{2}{*}{Date}&Load Variance&Load Factor&Area-load&Payment Minimization\\&Minimization&Maximization&Method& \\ \hline
July 27, 2017&26.4\%&30.6\%&14.4\%&8.5\% \\ \hline
July 6, 2017&18.6\%&26\%&15.5\%&2.9\% \\ \hline
July 23, 2017&15\%&23.1\%&14.1\%&7.3\%  \\ \hline
July 29, 2017&12.4\%&19.6\%&14\%&6.7\%  \\ \hline
average&18.1\%&24.8\%&14.5\%&6.4\%  \\ \hline \hline
\multicolumn{5}{|c|}{  \footnotesize  Load Factor Reduction With Respect to }\\
\multicolumn{5}{|c|}{  \footnotesize  Modified Load Variance Minimization Method}\\ \hline
\multirow{2}{*}{Date}&PODR&Load Factor&Area-load&Payment Minimization\\&Program&Maximization&Method& \\ \hline
July 27, 2017&-13.1\%&-4.1\%&-10.6\%&-24.5\% \\ \hline
July 6, 2017&-8.2\%&-3.8\%&-6.5\%&-23.6\% \\ \hline
July 23, 2017&-12.7\%&-4.4\%&-9.3\%&-31.4\%  \\ \hline
July 29, 2017&-8.4\%&-4.4\%&-4.7\%&-18.8\%  \\ \hline
average&-10.6\%&-4.2\%&-7.8\%&-24.6\%  \\ \hline
\end{tabular}
}
\end{table}

\section{Conclusion}
This paper investigated two critical aspects of power scheduling: residential users' energy costs and the power grid's load factor.
These two aspects should be considered simultaneously when utilities are to offer demand response programs to residential users.
In practice, minimizing the cost and maximizing the load factor are two conflicting objectives---one cannot be optimized without worsening the other. To address this challenge, we proposed a Pareto optimal demand response program that determined the Pareto optimal power scheduling. This program was constructed by solving a multiobjective optimization problem.
To explore and exploit the  decision space associated with the problem, a few algorithms were developed to generate feasible values for decision variables.
Relevant analysis was performed to show the effectiveness of the algorithms.
 Compared with existing methods for power scheduling, the Pareto optimal demand response program found a satisfying balance between the two objectives by sacrificing one objective slightly while substantially improving the other. This study was mainly focused on demand response for residential users. Our future work includes the investigation on multiobjective approaches to the design
of demand response programs for commercial or industrial customers and for vehicle-to-grid systems.

\appendix[Descriptions of Comparable Methods]

For modified area-load and payment minimization methods,
 the constraints in~(\ref{eq_MOP}) were addressed using the following constraint function:
\begin{equation*}
\begin{split}
 U(\bm{x})=   &   \sum_{c \in \mathcal{A}_{SF}}  \max \{ p_{c,\text{agr}}^{\min}  -         \sum_{h=s_c}^{e_c}     p_{c}^h , 0 \}                      \\
    & +    \max \{
B_{\text{ev}}^{\min}- B_{\text{ev}}^0-     \sum_{h=s_{\text{ev}}}^{e_{\text{ev}}}  p_{\text{ev}}^{h} \Delta h  ,0\}     \\
 & +   \max \{  B_{\text{ev}}^0+ \sum_{h=s_{\text{ev}}}^{e_{\text{ev}}}  p_{\text{ev}}^{h} \Delta h - B_{\text{ev}}^{\max} ,0 \} \\
 & +  \sum_{h \in \mathcal{H}} \max \{   B^h -B^{\max}  , - B^h,0 \}.  \\
\end{split}
\end{equation*}
If $U(\bm{x})>0$, then the power scheduling profile $\bm{x}$  violates the constraints and $\bm{x}$ is infeasible.

For the area-load method modified from\cite{13Kunwar}, we solved
\begin{equation}\label{eq_AL}
\begin{split}
  \underset{\bm{x}}{\min}  &\; f_{\text{cost}}(\bm{x})  + \sigma_1 U(\bm{x}) +\sigma_2 f_{\text{Penalty}}(\bm{x}) \\
    \underset{\bm{x}}{\max} & \; f_{\text{LF}}(\bm{x})- \sigma_1 U(\bm{x})\\
\end{split}
\end{equation}
where $\sigma_1$ and $\sigma_2$ represent the weighting factors of the constraint violation and penalty, respectively. Constraint handling was not discussed in~\cite{13Kunwar}, but it is essential because physical systems always induce constraints. We modified the objectives by including $U(\bm{x})$. The penalty of energy load deviating from the average load was evaluated as
\begin{equation}\label{eq_penalty}
 f_{\text{Penalty}}(\bm{x}) =   \sum_{h \in \mathcal{H} }  \mid E_{\text{total}}^h-E_{\text{av}}  \mid
\end{equation}
where
\begin{equation}\label{eq_Eav}
E_{\text{av}}=\sum_{h \in \mathcal{H} }E_{\text{total}}^h/H.
\end{equation}
The nondominated sorting genetic algorithm-II was used to solve~(\ref{eq_AL}).

For the payment minimization method modified from\cite{4914742}, we solved
\begin{equation}\label{eq_Pay}
  \underset{\bm{x}}{\min}   \; f_{\text{cost}}(\bm{x}) + \sigma_1 U(\bm{x})
\end{equation}
using particle swarm optimization.

The load factor maximization method  from\cite{11Sortomme} was modified as
\begin{equation}\label{eq_LF_max}
\begin{split}
    \underset{\bm{x}}{\max} & \; f_{\text{LF}}(\bm{x})  \\
 \text{subject to }   & (\ref{eq_cst_c}), (\ref{eq_cst_ev}), \text{ and } (\ref{eq_cst_storage})  \qquad  \forall c \in \mathcal{A}_{SF}.
\end{split}
\end{equation}
For the load variance minimization method modified from\cite{11Sortomme}, we solved
\begin{equation}\label{eq_LF_max}
\begin{split}
    \underset{\bm{x}}{\min} & \; \sum_{h \in \mathcal{H} }  ( E_{\text{total}}^h-E_{\text{av}}  )^2 /H   \\
 \text{subject to }   & (\ref{eq_cst_c}), (\ref{eq_cst_ev}), \text{ and } (\ref{eq_cst_storage})   \qquad  \forall c \in \mathcal{A}_{SF}
\end{split}
\end{equation}
where $E_{\text{av}}$ is defined in~(\ref{eq_Eav}).


\begin{thebibliography}{10}
\providecommand{\url}[1]{#1}
\csname url@samestyle\endcsname
\providecommand{\newblock}{\relax}
\providecommand{\bibinfo}[2]{#2}
\providecommand{\BIBentrySTDinterwordspacing}{\spaceskip=0pt\relax}
\providecommand{\BIBentryALTinterwordstretchfactor}{4}
\providecommand{\BIBentryALTinterwordspacing}{\spaceskip=\fontdimen2\font plus
\BIBentryALTinterwordstretchfactor\fontdimen3\font minus
  \fontdimen4\font\relax}
\providecommand{\BIBforeignlanguage}[2]{{%
\expandafter\ifx\csname l@#1\endcsname\relax
\typeout{** WARNING: IEEEtran.bst: No hyphenation pattern has been}%
\typeout{** loaded for the language `#1'. Using the pattern for}%
\typeout{** the default language instead.}%
\else
\language=\csname l@#1\endcsname
\fi
#2}}
\providecommand{\BIBdecl}{\relax}
\BIBdecl

\bibitem{5357331}
H.~Farhangi, ``The path of the smart grid,'' \emph{{IEEE} Power Energy Mag.},
  vol.~8, no.~1, pp. 18--28, Jan. 2010.

\bibitem{4787536}
A.~Ipakchi and F.~Albuyeh, ``Grid of the future,'' \emph{{IEEE} Power Energy
  Mag.}, vol.~7, no.~2, pp. 52--62, Mar. 2009.

\bibitem{6861946}
M.~Erol-Kantarci and H.~Mouftah, ``Energy-efficient information and
  communication infrastructures in the smart grid: A survey on interactions and
  open issues,'' \emph{{IEEE} Commun. Surveys Tuts.}, vol.~17, no.~1, pp.
  179--197, Jul. 2015.

\bibitem{7744778}
R.~Ma, H.~H. Chen, and W.~Meng, ``Dynamic spectrum sharing for the coexistence
  of smart utility networks and {WLANs} in smart grid communications,''
  \emph{{IEEE} Netw.}, vol.~31, no.~1, pp. 88--96, Jan./Feb. 2017.

\bibitem{7466137}
A.~Boustani, A.~Maiti, S.~Y. Jazi, M.~Jadliwala, and V.~Namboodiri, ``Seer
  grid: Privacy and utility implications of two-level load prediction in smart
  grids,'' \emph{{IEEE} Trans. Parallel Distrib. Syst.}, vol.~28, no.~2, pp.
  546--557, Feb. 2017.

\bibitem{7094306}
Y.~Wu, X.~Tan, L.~Qian, D.~H.~K. Tsang, W.~Song, and L.~Yu, ``Optimal pricing
  and energy scheduling for hybrid energy trading market in future smart
  grid,'' \emph{{IEEE} Trans. Ind. Informat.}, vol.~11, no.~6, pp. 1585--1596,
  Dec. 2015.

\bibitem{6376270}
W.-Y. Chiu, H.~Sun, and H.~V. Poor, ``Energy imbalance management using a
  robust pricing scheme,'' \emph{{IEEE} Trans. Smart Grid}, vol.~4, no.~2, pp.
  896--904, Jun. 2013.

\bibitem{6670131}
M.~Parvania, M.~Fotuhi-Firuzabad, and M.~Shahidehpour, ``Optimal demand
  response aggregation in wholesale electricity markets,'' \emph{{IEEE} Trans.
  Smart Grid}, vol.~4, no.~4, pp. 1957--1965, Dec. 2013.

\bibitem{6693775}
J.~H. Yoon, R.~Baldick, and A.~Novoselac, ``Dynamic demand response controller
  based on real-time retail price for residential buildings,'' \emph{{IEEE}
  Trans. Smart Grid}, vol.~5, no.~1, pp. 121--129, Jan. 2014.

\bibitem{6866904}
W.~Liu, Q.~Wu, F.~Wen, and J.~{\O}stergaard, ``Day-ahead congestion management
  in distribution systems through household demand response and distribution
  congestion prices,'' \emph{{IEEE} Trans. Smart Grid}, vol.~5, no.~6, pp.
  2739--2747, Nov. 2014.

\bibitem{6575202}
C.~Chen, J.~Wang, Y.~Heo, and S.~Kishore, ``{MPC}-based appliance scheduling
  for residential building energy management controller,'' \emph{{IEEE} Trans.
  Smart Grid}, vol.~4, no.~3, pp. 1401--1410, Sep. 2013.

\bibitem{6476768}
X.~Chen, T.~Wei, and S.~Hu, ``Uncertainty-aware household appliance scheduling
  considering dynamic electricity pricing in smart home,'' \emph{{IEEE} Trans.
  Smart Grid}, vol.~4, no.~2, pp. 932--941, Jun. 2013.

\bibitem{5540263}
A.-H. Mohsenian-Rad and A.~Leon-Garcia, ``Optimal residential load control with
  price prediction in real-time electricity pricing environments,''
  \emph{{IEEE} Trans. Smart Grid}, vol.~1, no.~2, pp. 120--133, 2010.

\bibitem{6316164}
S.~Salinas, M.~Li, and P.~Li, ``Multi-objective optimal energy consumption
  scheduling in smart grids,'' \emph{{IEEE} Trans. Smart Grid}, vol.~4, no.~1,
  pp. 341--348, Mar. 2013.

\bibitem{7093186}
H.~Lu, M.~Zhang, Z.~Fei, and K.~Mao, ``Multi-objective energy consumption
  scheduling in smart grid based on {Tchebycheff} decomposition,'' \emph{{IEEE}
  Trans. Smart Grid}, vol.~6, no.~6, pp. 2869--2883, Nov. 2015.

\bibitem{6462005}
P.~Yi, X.~Dong, A.~Iwayemi, C.~Zhou, and S.~Li, ``Real-time opportunistic
  scheduling for residential demand response,'' \emph{{IEEE} Trans. Smart
  Grid}, vol.~4, no.~1, pp. 227--234, Mar. 2013.

\bibitem{4652590}
B.~Ramanathan and V.~Vittal, ``A framework for evaluation of advanced direct
  load control with minimum disruption,'' \emph{{IEEE} Trans. Power Syst.},
  vol.~23, no.~4, pp. 1681--1688, Nov. 2008.

\bibitem{4914742}
M.~A.~A. Pedrasa, T.~D. Spooner, and I.~F. MacGill, ``Scheduling of demand side
  resources using binary particle swarm optimization,'' \emph{{IEEE} Trans.
  Power Syst.}, vol.~24, no.~3, pp. 1173--1181, Aug. 2009.

\bibitem{14Fadlullah}
Z.~M. Fadlullah, D.~M. Quan, N.~Kato, and I.~Stojmenovic, ``{GTES}: An
  optimized game-theoretic demand-side management scheme for smart grid,''
  \emph{{IEEE} Syst. J.}, vol.~8, no.~2, pp. 588--597, Jun. 2014.

\bibitem{10Mohsenian-Rad}
A.~H. Mohsenian-Rad, V.~W.~S. Wong, J.~Jatskevich, R.~Schober, and
  A.~Leon-Garcia, ``Autonomous demand-side management based on game-theoretic
  energy consumption scheduling for the future smart grid,'' \emph{{IEEE}
  Trans. Smart Grid}, vol.~1, no.~3, pp. 320--331, Dec. 2010.

\bibitem{15Eksin}
C.~Eksin, H.~Deli\c{c}, and A.~Ribeiro, ``Demand response management in smart
  grids with heterogeneous consumer preferences,'' \emph{{IEEE} Trans. Smart
  Grid}, vol.~6, no.~6, pp. 3082--3094, Nov. 2015.

\bibitem{12Nguyen}
H.~K. Nguyen, J.~B. Song, and Z.~Han, ``Demand side management to reduce
  peak-to-average ratio using game theory in smart grid,'' in \emph{Proc. IEEE
  INFOCOM Workshops}, Orlando, FL, Mar. 2012, pp. 91--96.

\bibitem{15Stephens}
E.~R. Stephens, D.~B. Smith, and A.~Mahanti, ``Game theoretic model predictive
  control for distributed energy demand-side management,'' \emph{{IEEE} Trans.
  Smart Grid}, vol.~6, no.~3, pp. 1394--1402, May 2015.

\bibitem{15Ma}
K.~Ma, G.~Hu, and C.~J. Spanos, ``A cooperative demand response scheme using
  punishment mechanism and application to industrial refrigerated warehouses,''
  \emph{{IEEE} Trans. Ind. Informat.}, vol.~11, no.~6, pp. 1520--1531, Dec.
  2015.

\bibitem{13Zhao}
Z.~Zhao, W.~C. Lee, Y.~Shin, and K.-B. Song, ``An optimal power scheduling
  method for demand response in home energy management system,'' \emph{{IEEE}
  Trans. Smart Grid}, vol.~4, no.~3, pp. 1391--1400, Sep. 2013.

\bibitem{J14}
H.-H. Chang, W.-Y. Chiu, and C.-M. Chen, ``User-centric multiobjective approach
  to privacy preservation and energy cost minimization in smart home,''
  \emph{{IEEE} Syst. J.}, vol.~13, no.~1, pp. 1030--1041, Mar. 2019.

\bibitem{18VENIZELOU}
V.~Venizelou, N.~Philippou, M.~Hadjipanayi, G.~Makrides, V.~Efthymiou, and
  G.~E. Georghiou, ``Development of a novel time-of-use tariff algorithm for
  residential prosumer price-based demand side management,'' \emph{Energy},
  vol. 142, pp. 633--646, Jan. 2018.

\bibitem{13Kunwar}
N.~Kunwar, K.~Yash, and R.~Kumar, ``Area-load based pricing in {DSM} through
  {ANN} and heuristic scheduling,'' \emph{{IEEE} Trans. Smart Grid}, vol.~4,
  no.~3, pp. 1275--1281, Sep. 2013.

\bibitem{16Kamyab}
F.~Kamyab, M.~Amini, S.~Sheykhha, M.~Hasanpour, and M.~M. Jalali, ``Demand
  response program in smart grid using supply function bidding mechanism,''
  \emph{{IEEE} Trans. Smart Grid}, vol.~7, no.~3, pp. 1277--1284, May 2016.

\bibitem{15Forouzandehmehr}
N.~Forouzandehmehr, M.~Esmalifalak, H.~Mohsenian-Rad, and Z.~Han, ``Autonomous
  demand response using stochastic differential games,'' \emph{{IEEE} Trans.
  Smart Grid}, vol.~6, no.~1, pp. 291--300, Jan. 2015.

\bibitem{1607Safdarian}
A.~{Safdarian}, M.~{Fotuhi-Firuzabad}, and M.~{Lehtonen}, ``Optimal residential
  load management in smart grids: {A} decentralized framework,'' \emph{{IEEE}
  Trans. Smart Grid}, vol.~7, no.~4, pp. 1836--1845, Jul. 2016.

\bibitem{15ZAREEN}
N.~Zareen, M.~W. Mustafa, U.~Sultana, R.~Nadia, and M.~A. Khattak, ``Optimal
  real time cost-benefit based demand response with intermittent resources,''
  \emph{Energy}, vol.~90, pp. 1695--1706, Oct. 2015.

\bibitem{16LiMa}
L.~Ma, N.~Liu, L.~Wang, J.~Zhang, J.~Lei, Z.~Zeng, C.~Wang, and M.~Cheng,
  ``Multi-party energy management for smart building cluster with {PV} systems
  using automatic demand response,'' \emph{Energy and Buildings}, vol. 121, pp.
  11--21, Jun. 2016.

\bibitem{17Li}
C.~{Li}, X.~{Yu}, W.~{Yu}, G.~{Chen}, and J.~{Wang}, ``Efficient computation
  for sparse load shifting in demand side management,'' \emph{{IEEE} Trans.
  Smart Grid}, vol.~8, no.~1, pp. 250--261, Jan. 2017.

\bibitem{14Safdarian}
A.~Safdarian, M.~Fotuhi-Firuzabad, and M.~Lehtonen, ``A distributed algorithm
  for managing residential demand response in smart grids,'' \emph{{IEEE}
  Trans. Ind. Informat.}, vol.~10, no.~4, pp. 2385--2393, Nov. 2014.

\bibitem{14Adika}
C.~O. {Adika} and L.~{Wang}, ``Demand-side bidding strategy for residential
  energy management in a smart grid environment,'' \emph{{IEEE} Trans. Smart
  Grid}, vol.~5, no.~4, pp. 1724--1733, Jul. 2014.

\bibitem{18NOOR}
S.~Noor, W.~Yang, M.~Guo, K.~H. van Dam, and X.~Wang, ``Energy demand side
  management within micro-grid networks enhanced by blockchain,'' \emph{Applied
  Energy}, vol. 228, pp. 1385--1398, Oct. 2018.

\bibitem{18Diekerhof}
M.~Diekerhof, F.~Peterssen, and A.~Monti, ``Hierarchical distributed robust
  optimization for demand response services,'' \emph{{IEEE} Trans. Smart Grid},
  vol.~9, no.~6, pp. 6018--6029, Nov. 2018.

\bibitem{19Luo}
F.~{Luo}, G.~{Ranzi}, C.~{Wan}, Z.~{Xu}, and Z.~Y. {Dong}, ``A multistage home
  energy management system with residential photovoltaic penetration,''
  \emph{{IEEE} Trans. Ind. Informat.}, vol.~15, no.~1, pp. 116--126, Jan. 2019.

\bibitem{16Pradhan}
V.~{Pradhan}, V.~S.~K. {Murthy Balijepalli}, and S.~A. {Khaparde}, ``An
  effective model for demand response management systems of residential
  electricity consumers,'' \emph{{IEEE} Syst. J.}, vol.~10, no.~2, pp.
  434--445, Jun. 2016.

\bibitem{15FARZAN}
F.~Farzan, M.~A. Jafari, J.~Gong, F.~Farzan, and A.~Stryker, ``A multi-scale
  adaptive model of residential energy demand,'' \emph{Applied Energy}, vol.
  150, pp. 258--273, Jul. 2015.

\bibitem{15Althaher}
S.~Althaher, P.~Mancarella, and J.~Mutale, ``Automated demand response from
  home energy management system under dynamic pricing and power and comfort
  constraints,'' \emph{{IEEE} Trans. Smart Grid}, vol.~6, no.~4, pp.
  1874--1883, Jul. 2015.

\bibitem{17Solanki}
B.~V. Solanki, A.~Raghurajan, K.~Bhattacharya, and C.~A. {Ca\~{n}izares},
  ``Including smart loads for optimal demand response in integrated energy
  management systems for isolated microgrids,'' \emph{{IEEE} Trans. Smart
  Grid}, vol.~8, no.~4, pp. 1739--1748, Jul. 2017.

\bibitem{16Yao}
E.~{Yao}, P.~{Samadi}, V.~W.~S. {Wong}, and R.~{Schober}, ``Residential demand
  side management under high penetration of rooftop photovoltaic units,''
  \emph{{IEEE} Trans. Smart Grid}, vol.~7, no.~3, pp. 1597--1608, May 2016.

\bibitem{15Paterakis}
N.~G. Paterakis, O.~{Erdin\c{c}}, A.~G. Bakirtzis, and J.~P.~S. {Catal\~{a}o},
  ``Optimal household appliances scheduling under day-ahead pricing and
  load-shaping demand response strategies,'' \emph{{IEEE} Trans. Ind.
  Informat.}, vol.~11, no.~6, pp. 1509--1519, Dec. 2015.

\bibitem{17Hayes}
B.~{Hayes}, I.~{Melatti}, T.~{Mancini}, M.~{Prodanovic}, and E.~{Tronci},
  ``Residential demand management using individualized demand aware price
  policies,'' \emph{{IEEE} Trans. Smart Grid}, vol.~8, no.~3, pp. 1284--1294,
  May 2017.

\bibitem{16MA}
K.~Ma, T.~Yao, J.~Yang, and X.~Guan, ``Residential power scheduling for demand
  response in smart grid,'' \emph{Elect. Power and Energy Syst.}, vol.~78, pp.
  320--325, Jun. 2016.

\bibitem{18Bahrami_PW}
S.~{Bahrami}, M.~H. {Amini}, M.~{Shafie-khah}, and J.~P.~S. {Catal\~{a}o}, ``A
  decentralized electricity market scheme enabling demand response
  deployment,'' \emph{{IEEE} Trans. Power Syst.}, vol.~33, no.~4, pp.
  4218--4227, Jul 2018.

\bibitem{VAZQUEZCANTELI20191072}
J.~R. {V\'{a}zquez-Canteli} and Z.~Nagy, ``Reinforcement learning for demand
  response: A review of algorithms and modeling techniques,'' \emph{Applied
  Energy}, vol. 235, pp. 1072--1089, Feb. 2019.

\bibitem{18Bahrami}
S.~Bahrami, V.~W.~S. Wong, and J.~Huang, ``An online learning algorithm for
  demand response in smart grid,'' \emph{{IEEE} Trans. Smart Grid}, vol.~9,
  no.~5, pp. 4712--4725, Sep. 2018.

\bibitem{18Dehghanpour}
K.~Dehghanpour, M.~H. Nehrir, J.~W. Sheppard, and N.~C. Kelly, ``Agent-based
  modeling of retail electrical energy markets with demand response,''
  \emph{{IEEE} Trans. Smart Grid}, vol.~9, no.~4, pp. 3465--3475, Jul. 2018.

\bibitem{16McKenna}
K.~McKenna and A.~Keane, ``Residential load modeling of price-based demand
  response for network impact studies,'' \emph{{IEEE} Trans. Smart Grid},
  vol.~7, no.~5, pp. 2285--2294, Sep. 2016.

\bibitem{1611Paterakis}
N.~G. {Paterakis}, I.~N. {Pappi}, O.~{Erdin\c{c}}, R.~{Godina}, E.~M.~G.
  {Rodrigues}, and J.~P.~S. {Catal\~{a}o}, ``Consideration of the impacts of a
  smart neighborhood load on transformer aging,'' \emph{{IEEE} Trans. Smart
  Grid}, vol.~7, no.~6, pp. 2793--2802, Nov. 2016.

\bibitem{16AGHAEI}
J.~Aghaei, M.-I. Alizadeh, P.~Siano, and A.~Heidari, ``Contribution of
  emergency demand response programs in power system reliability,''
  \emph{Energy}, vol. 103, pp. 688--696, May 2016.

\bibitem{16HAIDER}
H.~T. Haider, O.~H. See, and W.~Elmenreich, ``A review of residential demand
  response of smart grid,'' \emph{Renewable and Sustainable Energy Rev.},
  vol.~59, pp. 166--178, Jun. 2016.

\bibitem{16ESTHER}
B.~P. Esther and K.~S. Kumar, ``A survey on residential demand side management
  architecture, approaches, optimization models and methods,'' \emph{Renewable
  and Sustainable Energy Rev.}, vol.~59, pp. 342--351, Jul. 2016.

\bibitem{16Safdarian}
A.~{Safdarian}, M.~{Fotuhi-Firuzabad}, and M.~{Lehtonen}, ``Benefits of demand
  response on operation of distribution networks: {A} case study,''
  \emph{{IEEE} Syst. J.}, vol.~10, no.~1, pp. 189--197, Mar. 2016.

\bibitem{18Hansen}
T.~M. {Hansen}, E.~K.~P. {Chong}, S.~{Suryanarayanan}, A.~A. {Maciejewski}, and
  H.~J. {Siegel}, ``A partially observable markov decision process approach to
  residential home energy management,'' \emph{{IEEE} Trans. Smart Grid},
  vol.~9, no.~2, pp. 1271--1281, Mar. 2018.

\bibitem{18Shafie-Khah}
M.~{Shafie-Khah} and P.~{Siano}, ``A stochastic home energy management system
  considering satisfaction cost and response fatigue,'' \emph{{IEEE} Trans.
  Ind. Informat.}, vol.~14, no.~2, pp. 629--638, Feb. 2018.

\bibitem{18Rastegar}
M.~{Rastegar}, M.~{Fotuhi-Firuzabad}, and M.~{Moeini-Aghtai}, ``Developing a
  two-level framework for residential energy management,'' \emph{{IEEE} Trans.
  Smart Grid}, vol.~9, no.~3, pp. 1707--1717, May 2018.

\bibitem{14Vivekananthan}
C.~{Vivekananthan}, Y.~{Mishra}, G.~{Ledwich}, and F.~{Li}, ``Demand response
  for residential appliances via customer reward scheme,'' \emph{{IEEE} Trans.
  Smart Grid}, vol.~5, no.~2, pp. 809--820, Mar. 2014.

\bibitem{18Zhou}
B.~{Zhou}, R.~{Yang}, C.~{Li}, Y.~{Cao}, Q.~{Wang}, and J.~{Liu},
  ``Multiobjective model of time-of-use and stepwise power tariff for
  residential consumers in regulated power markets,'' \emph{{IEEE} Syst. J.},
  vol.~12, no.~3, pp. 2676--2687, Sep. 2018.

\bibitem{19Ghasemkhani}
A.~{Ghasemkhani}, L.~{Yang}, and J.~{Zhang}, ``Learning-based demand response
  for privacy-preserving users,'' \emph{{IEEE} Trans. Ind. Informat.}, 2019,
  early access.

\bibitem{15Wen}
Z.~{Wen}, D.~{O'Neill}, and H.~{Maei}, ``Optimal demand response using
  device-based reinforcement learning,'' \emph{{IEEE} Trans. Smart Grid},
  vol.~6, no.~5, pp. 2312--2324, Sep. 2015.

\bibitem{Ruelens17}
F.~{Ruelens}, B.~J. {Claessens}, S.~{Vandael}, B.~{De Schutter},
  R.~{Babu\v{s}ka}, and R.~{Belmans}, ``Residential demand response of
  thermostatically controlled loads using batch reinforcement learning,''
  \emph{{IEEE} Trans. Smart Grid}, vol.~8, no.~5, pp. 2149--2159, Sep. 2017.

\bibitem{16Paterakis}
N.~G. {Paterakis}, A.~{Ta\c{s}c{\i}karao\u{g}lu}, O.~{Erdin\c{c}}, A.~G.
  {Bakirtzis}, and J.~P.~S. {Catal\~{a}o}, ``Assessment of
  demand-response-driven load pattern elasticity using a combined approach for
  smart households,'' \emph{{IEEE} Trans. Ind. Informat.}, vol.~12, no.~4, pp.
  1529--1539, Aug. 2016.

\bibitem{18wang}
F.~{Wang}, K.~{Li}, C.~{Liu}, Z.~{Mi}, M.~{Shafie-Khah}, and J.~P.~S.
  {Catal\~{a}o}, ``Synchronous pattern matching principle-based residential
  demand response baseline estimation: Mechanism analysis and approach
  description,'' \emph{{IEEE} Trans. Smart Grid}, vol.~9, no.~6, pp.
  6972--6985, Nov. 2018.

\bibitem{19Elghitani}
F.~{Elghitani} and E.~{El-Saadany}, ``Smoothing net load demand variations
  using residential demand management,'' \emph{{IEEE} Trans. Ind. Informat.},
  vol.~15, no.~1, pp. 390--398, Jan. 2019.

\bibitem{19Ma}
K.~{Ma}, Y.~{Yu}, B.~{Yang}, and J.~{Yang}, ``Demand-side energy management
  considering price oscillations for residential building heating and
  ventilation systems,'' \emph{{IEEE} Trans. Ind. Informat.}, 2019, early
  access.

\bibitem{17Erdinc}
O.~{Erdin\c{c}}, A.~{Ta\c{s}c{\i}karao\u{g}lu}, N.~G. {Paterakis}, Y.~{Eren},
  and J.~P.~S. {Catal\~{a}o}, ``End-user comfort oriented day-ahead planning
  for responsive residential {HVAC} demand aggregation considering weather
  forecasts,'' \emph{{IEEE} Trans. Smart Grid}, vol.~8, no.~1, pp. 362--372,
  Jan. 2017.

\bibitem{17Mahdavi}
N.~{Mahdavi}, J.~H. {Braslavsky}, M.~M. {Seron}, and S.~R. {West}, ``Model
  predictive control of distributed air-conditioning loads to compensate
  fluctuations in solar power,'' \emph{{IEEE} Trans. Smart Grid}, vol.~8,
  no.~6, pp. 3055--3065, Nov. 2017.

\bibitem{18Rassaei}
F.~{Rassaei}, W.~{Soh}, and K.~{Chua}, ``Distributed scalable autonomous
  market-based demand response via residential plug-in electric vehicles in
  smart grids,'' \emph{{IEEE} Trans. Smart Grid}, vol.~9, no.~4, pp.
  3281--3290, Jul. 2018.

\bibitem{18Pal}
S.~{Pal} and R.~{Kumar}, ``Electric vehicle scheduling strategy in residential
  demand response programs with neighbor connection,'' \emph{{IEEE} Trans. Ind.
  Informat.}, vol.~14, no.~3, pp. 980--988, Mar. 2018.

\bibitem{18Munshi}
A.~A. {Munshi} and Y.~A.~I. {Mohamed}, ``Extracting and defining flexibility of
  residential electrical vehicle charging loads,'' \emph{{IEEE} Trans. Ind.
  Informat.}, vol.~14, no.~2, pp. 448--461, Feb. 2018.

\bibitem{7004894}
C.~Perera, C.~H. Liu, S.~Jayawardena, and M.~Chen, ``A survey on {I}nternet of
  {T}hings from industrial market perspective,'' \emph{{IEEE} Access}, vol.~2,
  pp. 1660--1679, Jan. 2015.

\bibitem{6935003}
C.~H. Liu, J.~Fan, J.~W. Branch, and K.~K. Leung, ``Toward {QoI} and
  energy-efficiency in {Internet-of-Things} sensory environments,''
  \emph{{IEEE} Trans. Emerg. Topics Comput.}, vol.~2, no.~4, pp. 473--487, Dec.
  2014.

\bibitem{6461496}
C.~Chen, K.~Nagananda, G.~Xiong, S.~Kishore, and L.~Snyder, ``A
  communication-based appliance scheduling scheme for consumer-premise energy
  management systems,'' \emph{{IEEE} Trans. Smart Grid}, vol.~4, no.~1, pp.
  56--65, Mar. 2013.

\bibitem{6674101}
K.~N. Kumar, B.~Sivaneasan, P.~H. Cheah, P.~L. So, and D.~Z.~W. Wang, ``{V2G}
  capacity estimation using dynamic {EV} scheduling,'' \emph{{IEEE} Trans.
  Smart Grid}, vol.~5, no.~2, pp. 1051--1060, Mar. 2014.

\bibitem{7927719}
L.~Park, Y.~Jang, S.~Cho, and J.~Kim, ``Residential demand response for
  renewable energy resources in smart grid systems,'' \emph{{IEEE} Trans. Ind.
  Informat.}, vol.~13, no.~6, pp. 3165--3173, Dec. 2017.

\bibitem{14Gonen}
T.~G\"{o}nen, \emph{Electric Power Distribution Engineering}, 3rd~ed.\hskip 1em
  plus 0.5em minus 0.4em\relax Boca Raton, FL, USA: CRC Press, Jan. 2014.

\bibitem{7050260}
W.-Y. Chiu, H.~Sun, and H.~Vincent~Poor, ``A multiobjective approach to
  multimicrogrid system design,'' \emph{{IEEE} Trans. Smart Grid}, vol.~6,
  no.~5, pp. 2263--2272, Sep. 2015.

\bibitem{J13}
W.-Y. Chiu, ``Method of reduction of variables for bilinear matrix inequality
  problems in system and control designs,'' \emph{{IEEE} Trans. Syst., Man,
  Cybern., Syst.}, vol.~47, no.~7, pp. 1241--1256, Jul. 2017.

\bibitem{17Chong}
E.~K.~P. Chong and S.~H. \.{Z}ak, \emph{An Introduction to Optimization},
  4th~ed.\hskip 1em plus 0.5em minus 0.4em\relax Hoboken, NJ, USA: Wiley, 2013.

\bibitem{MOEA_bk1}
C.~A. {Coello Coello}, D.~A. {Van Veldhuizen}, and G.~B. Lamont,
  \emph{Evolutionary Algorithms for Solving Multi-Objective Problems}.\hskip
  1em plus 0.5em minus 0.4em\relax Boston, MA, USA: Springer Science{+}Business
  Media, LLC, 2007.

\bibitem{15Zhang}
X.~Zhang, Y.~Tian, and Y.~Jin, ``A knee point-driven evolutionary algorithm for
  many-objective optimization,'' \emph{{IEEE} Trans. Evol. Comput.}, vol.~19,
  no.~6, pp. 761--776, Dec. 2015.

\bibitem{09Rachmawati}
L.~{Rachmawati} and D.~{Srinivasan}, ``Multiobjective evolutionary algorithm
  with controllable focus on the knees of the {P}areto front,'' \emph{{IEEE}
  Trans. Evol. Comput.}, vol.~13, no.~4, pp. 810--824, Aug. 2009.

\bibitem{7465803}
W.-Y. Chiu, G.~G. Yen, and T.~K. Juan, ``Minimum {M}anhattan distance approach
  to multiple criteria decision making in multiobjective optimization
  problems,'' \emph{{IEEE} Trans. Evol. Comput.}, vol.~20, no.~6, pp. 972--985,
  Dec. 2016.

\bibitem{EIA}
\BIBentryALTinterwordspacing
{U.S.} {E}nergy {I}nformation {A}dministratino (eia), {A}ccessed on {J}an. 9,
  2019. [Online]. Available:
  \url{https://www.eia.gov/totalenergy/data/monthly/}
\BIBentrySTDinterwordspacing

\bibitem{Web1}
\BIBentryALTinterwordspacing
Charging speeds \& connectors. {A}ccessed on Sep. 1, 2018. [Online]. Available:
  \url{https://www.zap-map.com/charge-points/connectors-speeds/}
\BIBentrySTDinterwordspacing

\bibitem{5991951}
M.~A. Alahmad, P.~G. Wheeler, A.~Schwer, J.~Eiden, and A.~Brumbaugh, ``A
  comparative study of three feedback devices for residential real-time energy
  monitoring,'' \emph{{IEEE} Trans. Ind. Electron.}, vol.~59, no.~4, pp.
  2002--2013, Apr. 2012.

\bibitem{Web2}
\BIBentryALTinterwordspacing
{PJM}. {A}ccessed on Sep. 1, 2018. [Online]. Available:
  \url{http://www.pjm.com}
\BIBentrySTDinterwordspacing

\bibitem{11Sortomme}
E.~Sortomme, M.~M. Hindi, S.~D.~J. MacPherson, and S.~S. Venkata, ``Coordinated
  charging of plug-in hybrid electric vehicles to minimize distribution system
  losses,'' \emph{{IEEE} Trans. Smart Grid}, vol.~2, no.~1, pp. 198--205, Mar.
  2011.

\end{thebibliography}
\end{document}